\documentclass[copyright,creativecommons]{eptcs}
\providecommand{\event}{EXPRESS/SOS 2020}
\usepackage{breakurl}             % Not needed if you use pdflatex only.
\usepackage{amsmath}
\usepackage{amsthm}
\usepackage{amssymb}
\usepackage[capitalise]{cleveref}
\usepackage{extarrows}    
\usepackage{amsfonts,mathtools}   
\usepackage{verbatim}    
\usepackage{tikz}
\newcommand{\imAppendix}[1]{}
\title{Correctly Implementing Synchronous Message Passing in the Pi-Calculus By Concurrent Haskell's MVars}
\author{Manfred Schmidt-Schau{\ss} % \thanks{supported by the Deutsche Forschungsgemeinschaft (DFG) under grant SCHM 986/11-1.}
\institute{Goethe-University, Frankfurt am Main, Germany}
\email{schauss@ki.cs.uni-frankfurt.de}
\and 
 David Sabel %%\footnotemark[1]
\institute{LMU, Munich, Germany} 
\email{david.sabel@ifi.lmu.de}
\thanks{This research is supported by the Deutsche Forschungsgemeinschaft (DFG) under grant SA2908/3-1}
}
\def\titlerunning{Implementing Synchronous Message-Passing by MVars}
\def\authorrunning{M. Schmidt-Schau{\ss} \& D. Sabel}

\newcommand{\citet}[1]{{\cite{#1}}}

\newenvironment{proof*}{{\it Proof.}}{}
\newtheorem{theorem}{Theorem}[section]
\newtheorem{lemma}[theorem]{Lemma} 
\newtheorem{example}[theorem]{Example} 
\newtheorem{definition}[theorem]{Definition} 
\newtheorem{remark}[theorem]{Remark} 
\newtheorem{proposition}[theorem]{Proposition}

\newtheorem{conjecture}[theorem]{Conjecture}

\begin{document}
%%%% Macros
%%% TEXEXPAND: INCLUDED FILE MARKER ./pitrans-macros.tex
\newcommand{\inputxy}[2]{#1(#2)}
\newcommand{\outputxy}[2]{\overline{#1}#2}
\newcommand{\tauNullInduced}[1][T]{\phi_{0,#1}}
\newcommand{\tauInduced}[1][T]{\phi_{#1}}
\newcommand{\sigmaNullInduced}[1][T]{\psi_{0,#1}}
\newcommand{\sigmaInduced}[1][T]{\psi_{#1}}
\newcommand{\putCh}{\ensuremath{\mathtt{putC}}}
\newcommand{\putS}{\ensuremath{\mathtt{putS}}} 
\newcommand{\takeS}{\ensuremath{\mathtt{takeS}}}
\newcommand{\takeCh}{\ensuremath{\mathtt{takeC}}}
% %%\newcommand{\typetau}{{\color{red}{\mathfrak{t}}}}
\newcommand{\typetau}{{\mathfrak{t}}}
\newcommand{\barb}[1]{\,{\Rsh}^{#1}}
\newcommand{\maybarb}{\!\downharpoonright}
\newcommand{\maydivbarb}{\!\upharpoonright}
\newcommand{\shouldbarb}{\downharpoonright\!\downharpoonright}
\newcommand{\mustdivbarb}{\upharpoonright\!\upharpoonright}
\newcommand{\leqmaybarb}{\leq_{c,\maybarb_x}}
\newcommand{\leqshouldbarb}{\leq_{c,\shouldbarb_x}}
\newcommand{\leqbarb}{\leq_{c,\mathit{barb}}}
\newcommand{\simbarb}{\sim_{c,\mathit{barb}}}
\newcommand{\cplambda}{ \texttt{cp}\lambda}
\newcommand{\noncp}{\mathit{noncp}} 
\newcommand{\rl}{\mathit{rl}}
\newcommand{\PT}{\mathit{PT}}
\newcommand{\wrt}{w.r.t.\,}
\newcommand{\eg}{e.g.}
\newcommand{\ie}{i.e.}
\newcommand{\SR}{\ensuremath{\mathit{SR}}}
\newcommand{\NSR}{\ensuremath{\mathit{NSR}}}
\newcommand{\infsr}{\ensuremath{{\mathit{infSR}}}}
\newcommand{\tBot}{{\tt Bot}}
\newcommand{\ITT}{{\mathit{IT}}}
\newcommand{\italt}{{\mathit{alt}}}
\newcommand{\tletrx}[2]{(\tletrec~#1 ~{\tt in}~#2)}
\newcommand{\iRED}{{\mathit{Red}}}
\newcommand{\RCtxt}{\mathbb{R}}
\newcommand{\itrace}{{\mathit{trace}}}
\newcommand{\nullproc}{{\mathbf{0}}}
\newcommand{\PAR}{{\ensuremath{\hspace{0.1mm}{\!\scalebox{2}[1]{\tt |}\!}\hspace{0.1mm}}}}
\newcommand{\OR}{~|~}
\newcommand{\MVAR}[2]{{#1}\,{\normalfont\textbf{\textsf{m}}}\,{#2}}
\newcommand{\EMPTYMVAR}[1]{\MVAR{{#1}}{-}}
\newcommand{\Proc}{\textit{Proc}}
\newcommand{\ProcCH}{{\ensuremath{\textit{Proc}_{\CH}}}}
\newcommand{\ProcCHF}{{\ensuremath{\textit{Proc}_{\CHF}}}}
\newcommand{\ExprCH}{{\ensuremath{\textit{Expr}_{\CH}}}}
\newcommand{\ExprCHF}{{\ensuremath{\textit{Expr}_{\CHF}}}}
\newcommand{\Expr}{\textit{Expr}}
\newcommand{\Typ}{\textit{Typ}}
\newcommand{\TypCH}{{\ensuremath{\textit{Typ}_{\CH}}}}
\newcommand{\TypCHF}{{\ensuremath{\textit{Typ}_{\CHF}}}}
\newcommand{\IExpr}{\textit{IExpr}}
\newcommand{\MExpr}{\textit{MExpr}}
\newcommand{\MExprCH}{{\ensuremath{\textit{MExpr}_{\CH}}}}
\newcommand{\MExprCHF}{{\ensuremath{\textit{MExpr}_{\CHF}}}}
\newcommand{\IMExpr}{\textit{IMExpr}}
\newcommand{\Var}{\textit{Var}}
\newcommand{\Id}{I}
\newcommand{\ignore}[1]{}
\newcommand{\NEW}{\nu}
\newcommand{\THREAD}[2]{{#1}\,{\Leftarrow}\,{#2}}
\newcommand{\LAZYTHREAD}[2]{{#1}\,{\stackrel{\mathit{lazy}}{\Leftarrow}}\,{#2}}
\newcommand{\myvdash}{\hspace*{.5pt}{\vdash}\hspace*{.5pt}}
\newcommand{\mydcolon}{\hspace*{.5pt}{::}\hspace*{.5pt}}
\newcommand{\myquad}{\!\quad\!}
\newcommand{\SHARE}[2]{{#1}={#2}}
\newcommand{\arity}{\mathrm{ar}}
\newcommand{\tletr}{\text{\normalfont\ttfamily letrec}}
\newcommand{\tletrec}{\text{\normalfont\ttfamily letrec}}
\newcommand{\tof}{\text{{\normalfont\ttfamily of}}}
\newcommand{\tTrue}{\text{{\normalfont\ttfamily True}}}
\newcommand{\tFalse}{\text{{\normalfont\ttfamily False}}}
\newcommand{\tin}{\text{{\normalfont\ttfamily in}}}
\newcommand{\tseq}{\text{{\normalfont\ttfamily seq}}}
\newcommand{\tcase}{\text{{\normalfont\ttfamily case}}}
\newcommand{\treturn}{\text{\normalfont\ttfamily return}}
\newcommand{\tbind}{{\,\text{\normalfont{\ensuremath{\texttt{\char062\!\char062=}}}}}\,}
\newcommand{\tbindthen}{{\,\text{\normalfont{\ensuremath{\texttt{\char062\!\char062}}}}}\,}
\newcommand{\tforkIO}{\text{{\normalfont\ttfamily forkIO}}}
\newcommand{\tfuture}{\text{{\normalfont\ttfamily future}}}
\newcommand{\ttakeMVar}{\text{{\normalfont\ttfamily takeMVar}}}
\newcommand{\tnewMVar}{\text{{\normalfont\ttfamily newMVar}}}
\newcommand{\tnewEmptyMVar}{\text{{\normalfont\ttfamily newEmptyMVar}}}
\newcommand{\tputMVar}{\text{{\normalfont\ttfamily putMVar}}}
\newcommand{\FV}{\mathit{FV}}
\newcommand{\FN}{\mathit{FN}}
\newcommand{\BN}{\mathit{BN}}
\newcommand{\MTHREAD}[2]{{#1} \xLongleftarrow{\!\!\text{\normalfont\sffamily main}\!\!}{#2}}
\newcommand{\PCtxtsCH}{{\ensuremath{\textit{PCtxt}_{\CH}}}}
\newcommand{\PCtxtsCHF}{{\ensuremath{\textit{PCtxt}_{\CHF}}}}
\newcommand{\PCtxts}{\textit{PCtxt}}
\newcommand{\PCtxt}{\mathbb{D}}
\newcommand{\CtxtsCH}{{\ensuremath{\textit{Ctxt}_{\CH}}}}
\newcommand{\CtxtsCHF}{{\ensuremath{\textit{Ctxt}_{\CHF}}}}
\newcommand{\Ctxts}{\textit{CCtxt}}
\newcommand{\MCtxtsCH}{{\ensuremath{\textit{MCtxt}_{\CH}}}}
\newcommand{\MCtxtsCHF}{{\ensuremath{\textit{MCtxt}_{\CHF}}}}
\newcommand{\MCtxts}{\textit{MCtxt}}
\newcommand{\ECtxtsCH}{{\ensuremath{\textit{ECtxt}_{\CH}}}}
\newcommand{\ECtxtsCHF}{{\ensuremath{\textit{ECtxt}_{\CHF}}}}
\newcommand{\ECtxts}{\textit{ECtxt}}
\newcommand{\FCtxtsCH}{{\ensuremath{\textit{FCtxt}_{\CH}}}}
\newcommand{\FCtxtsCHF}{{\ensuremath{\textit{FCtxt}_{\CHF}}}}
\newcommand{\FCtxts}{\textit{FCtxt}}
\newcommand{\Ctxt}{\mathbb{C}}
\newcommand{\TCtxt}{\mathbb{T}}
\newcommand{\TCtxts}{\textit{TCtxt}}
\newcommand{\SCtxt}{\mathbb{S}}
\newcommand{\SCtxts}{\textit{SCtxt}}
\newcommand{\MultiCtxt}{\widetilde{\mathbb{C}}}
\newcommand{\MCtxt}{\mathbb{M}}
\newcommand{\ECtxt}{\mathbb{E}}
\newcommand{\FCtxt}{\mathbb{F}}
\newcommand{\LCtxt}{\mathbb{L}}
\newcommand{\LCtxts}{\textit{LCtxt}}
\newcommand{\ACtxts}{\textit{ACtxts}}
\newcommand{\LHatCtxts}{\textit{$\widehat{LCtxt}$}}
\newcommand{\ACtxt}{\mathbb{A}}
\newcommand{\KCtxt}{\mathbb{K}}
\newcommand{\iEnv}{\mathit{Env}}
\newcommand{\reduce}[1][]{\xrightarrow{sr{#1}}}
\newcommand{\maycon}{{\downarrow}}
\newcommand{\mustcon}{{\Downarrow}}
\newcommand{\mustdiv}{{\Uparrow}}
\newcommand{\maydiv}{{\uparrow}}
\newcommand{\maycontau}{{\downarrow}_0}
\newcommand{\mustcontau}{{\Downarrow}_0}
\newcommand{\mustdivtau}{{\Uparrow}_0}
\newcommand{\maydivtau}{{\uparrow}_0}
\newcommand{\tundefined}{\text{{\normalfont\ttfamily undefined}}}
\newcommand{\iRed}{\mathit{Red}}
\newcommand{\ialts}{\mathit{alts}}
\newcommand{\ttmvar}{\text{tmvar}}
\newcommand{\tpmvar}{\text{pmvar}}
\newcommand{\srnr}{\mathit{srnr}}
\newcommand{\srnrp}{\mathit{srnrp}}
\newcommand{\tif}{\texttt{if}}
\newcommand{\tthen}{\texttt{then}}
\newcommand{\telse}{\texttt{else}}
\newcommand{\LR}{\mathit{LR}}
\newcommand{\LRP}{\mathit{LRP}}
\newcommand{\CHF}{\mathit{CHF}}
\newcommand{\CH}{{\mathit{CH}}}
\newcommand{\T}{\mathbb{T}}
\newcommand{\Red}{\mathit{Red}}
\newcommand{\cselambda}{ \mathrm{cse}\lambda}
\newcommand{\LV}{\mathit{LV}}
\newcommand{\calP}{\mathcal{P}}
\newcommand{\calC}{\mathcal{C}}
\newcommand{\CHFST}{\error CHFST in diesem Papier nicht, da callbyname CHF\ensuremath{\mathit{CHF^*}}}
%%%%%%%
\newcommand{\tDummyChannel}{\text{{\normalfont\ttfamily DummyChannel}}}
\newcommand{\tChannel}{\text{{\normalfont\ttfamily Channel}}}
\newcommand{\tChan}{\text{{\normalfont\ttfamily Chan}}}
\newcommand{\LHatCtxt}{\textit{$\widehat{\LCtxt}$}}
\newcommand{\sr}{\xrightarrow{sr}}
\newcommand{\ia}{\xrightarrow{ia}}
\newcommand{\icode}[1]{\textit{#1}}
\newcommand{\tgetsend}{\text{{\normalfont\ttfamily getsend}}}
\newcommand{\tgetcheck}{\text{{\normalfont\ttfamily getcheck}}}
\newcommand{\tStop}{\text{{\normalfont\ttfamily Stop}}}
\newcommand{\tstop}{\text{{\normalfont\ttfamily stop}}}
\newcommand{\bfdo}{\ensuremath{\mathbf{do}}}
\newcommand{\casepf}{\,\texttt{->}\,}
\newcommand{\tMVar}{\texttt{MVar}}
\newcommand{\tsend}{\texttt{send}}
\newcommand{\tsendx}{\texttt{sendx}}
\newcommand{\tsendz}{\texttt{sendz}}
\newcommand{\tcheck}{\texttt{check}}
\newcommand{\tcheckx}{\texttt{checkx}}
\newcommand{\tcheckz}{\texttt{checkz}}
\newcommand{\icheckx}{\mathit{checkx}}
\newcommand{\isendx}{\mathit{sendx}}
\newcommand{\istop}{{\ensuremath{\mathit{stop}}}}    %\color{green!50!black}
\newcommand{\tlet}{\texttt{let}}
\newcommand{\TODO}[1]{~\\ {\bf TODO:    #1}~\\ }
\newcommand{\pistop}{\Pi_{\tStop}}
\newcommand{\pistopCc}{\Pi_{\tStop,C}^c}
\newcommand{\pistopC}{\Pi_{\tStop,C}}
\newcommand{\picalc}{\Pi}
 \newcommand{\UTHREAD}[1]{{\Leftarrow}\,{#1}}
\newcommand{\UMTHREAD}[1]{{\xLongleftarrow{\!\!\text{\normalfont\sffamily main}\!\!}}\,{#1}}
\newcommand{\Coutsigma}{C_{\mathit{out}}^\sigma}
\newcommand{\Couttau}{C_{\mathit{out}}^\tau}
%%%%%%%%%%%%%%%%%%%%%%%%%%%%%%%%%%%%
%%
%\ignore{
%%black, blue, brown, cyan, darkgray, gray, green, lightgray, lime, magenta, olive, orange, pink, purple, red, teal, violet, white, yellow.

% \newcommand{\manfred}[1]{\color{magenta}Manfred: {#1}\color{black}} 
% \newcommand{\manfredins}[1]{{\textcolor{magenta}{{#1}}}}
% \newcommand{\manfreddel}[1]{{\textcolor{magenta}{Some phrases deleted!!}}}
% \newcommand{\manfredcomment}[1]{{~\\ \textcolor{magenta}{Manfred: {#1}}~\\}}
%%\newcommand{\manfred}[1]{{~\\ \textcolor{magenta}{Manfred: {#1}}~\\}}
% \newcommand{\david}[1]{\color{blue}David: {#1}\color{black}}
% \newcommand{\davidins}[1]{\color{blue}{#1}\color{black}}
% \newcommand{\davidrm}[1]{\color{black!50!white}{{\bf loeschen (DS):#1}}\color{black}}
% \newcommand{\loeschen}[1]{\color{black!50!white}{{\bf loeschen:} #1}\color{black}}
%}
\newcommand{\reportOnly}[1]{}
\newcommand{\paperOnly}[1]{{#1}} 
\maketitle

\begin{abstract}
%\input{abstract-CH}
%%% TEXEXPAND: INCLUDED FILE MARKER ./abstract-CHWS.tex
Comparison of concurrent programming languages and  correctness of program transformations in concurrency are the focus of this research.
As criterion we use contextual semantics adapted to concurrency, where may- as well as should-convergence are observed.
We investigate the relation between the synchronous pi-calculus and a core language of Concurrent Haskell (CH). 
The contextual semantics is on the one hand forgiving
with respect to the details of the operational semantics, and on the other hand
implies strong requirements for the interplay between the processes after translation.  
Our result is that CH embraces the synchronous pi-calculus.   
Our main task is to find and prove correctness of encodings of pi-calculus channels by CH's concurrency primitives, which are  MVars. 
They behave like (blocking) 1-place buffers modelling the shared-memory. 
The first developed translation uses an extra private MVar
for every communication.
We also automatically generate 
and check potentially correct translations that reuse the MVars 
where
one MVar contains the message and two additional MVars for synchronization are
used to model the synchronized communication of a single channel in the pi-calculus.    
Our automated experimental results 
lead to the conjecture that one additional MVar is insufficient.
%%% TEXEXPAND: END FILE ./abstract-CHWS.tex
\end{abstract}
% 
%  \keywords{pi-calculus, functional programming; concurrency, adequate translations}
% {\bf Keywords}: Expressivity of programming languages, concurrency, pi-calculus, functional programming,  adequate translations 
% 
%%% TEXEXPAND: INCLUDED FILE MARKER ./pitrans-text.tex
%%% TEXEXPAND: INCLUDED FILE MARKER ./intro-CHWS.tex
\section{Introduction}
Our goals are the comparison of programming languages, correctness of transformations, compilation and optimization of programs, 
in particular of concurrent programs.   
We already used the contextual semantics of concurrent (functional) programming languages to 
effectively verify correctness of transformations
\cite{niehren-sabel-schmidt-schauss-schwinghammer:07:entcs,sabel-schmidt-schauss-PPDP:2011,sabel-schmidt-schauss-LICS:12},
also under the premise not to worsen the runtime \cite{schauss-sabel-dallmeyer:18}.   
We propose to test may- and should-convergence in the contextual semantics, since, in particular, it rules out transformations that transform an always successful process into 
a process that  
may run into an error, for example a deadlock.
There are also other notions of program equivalence in the literature, like bisimulation based program equivalences \cite{sangiorgi-walker:01}, however, 
 these tend to take also  implementation specific   % \marginpar{non-relevant klang so negativ}    ok
 behavior 
 of the operational semantics into account, 
 whereas  contextual equivalence abstracts from the details of the executions.  

 In \cite{schmidt-schauss-niehren-schwinghammer-sabel-ifip-tcs:08,schmidtschauss-sabel-niehren-schwing-tcs:15} we developed notions of correctness of translations {w.r.t.}~contextual 
semantics,   %%and gave several (introductory) examples.
and in \cite{schwinghammer-sabel-schmidt-schauss-niehren:09:ml} we applied them in the context of concurrency, but for quite similar source and target languages.
% translation are quite similar.
% (they are variations of the same concurrent call-by-value lambda calculus with shared memory). 
% To underpin the techniques and to take a next step in our research we were looking for more interesting examples where 
% source and target language are both concurrent but also quite different. 
In this paper we translate a synchronous message passing model into 
% \davidins{***}\marginpar{\davidins{DS: asynchronous gelöscht, da es synchronous / asynchronous memory in anderer Verwendung gibt 
%  (PPDP referee war da glaub ich verwirrt}}
%%an (asynchronous) 
a shared memory model, namely 
 a synchronous $\pi$-calculus into a core-language of Concurrent Haskell, called $\CH$.

%%% Hier contextual semantics
The contextual semantics of concurrent programming languages is a generalization of the extensionality principle of functions. 
% a black-box testing principle of programs: the semantic equality of two (programmed) functions $f,g$ holds, if for all inputs 
% $n$, $f(n)$ and $g(n)$ compute the same output. 
% It is known that 
%\davidins{
The test for a program $P$ is whether
% $C[P]$ for all program-contexts $C$ 
$C[P]$ -- i.e. $P$ plugged into a program context -- 
successfully terminates (converges) or not, which usually means that
the standard reduction sequence ends with a value.
For a concurrent program $P$, we use two observations: {\em may-convergence} ($P\maycon$) -- at least one execution path terminates successfully, 
and {\em should-convergence}\footnote{An alternative 
% nondeterministic 
observation is must-convergence (all execution paths terminate). The advantages of equivalence notions based on may- and should-convergence are invariance under fairness restrictions, 
preservation of deadlock-freedom, 
and equivalence of busy-wait
and wait-until behavior (see \eg~\cite{schwinghammer-sabel-schmidt-schauss-niehren:09:ml}). 
} ($P\mustcon$) -- every intermediate state of a reduction sequence may-converges. 
For two processes $P$ and $Q$, $P \leq_c Q$ holds iff for all contexts $C[\cdot]$: $(C[P]\maycon \implies C[Q]\maycon)$, %%%  and  $(C[P]\mustcon \implies C[Q]\mustcon)$, 
and
$P$ and $Q$ are contextually equivalent, $P \sim_c Q$,  iff $P \leq_c Q$ and $Q \leq_c P$.   
Showing equal expressivity of two (concurrent) calculi by a translation $\tau$ %%%is founded on properties  of $\tau$: 
%%% $\tau$ leaves may- and should-convergence invariant.
%%% this 
requires that may- and should-convergence make sense in each calculus.
Important properties are convergence-equivalence (may- and should-convergencies are preserved and reflected by the translation) and  adequacy 
 (see  \cref{def:adequate-and-fullyabstract}), 
which holds if $\tau(P) \leq_{c,CH} \tau(Q)$ implies  $P \leq_{c,\pi} Q$,  for all $\pi$-calculus processes $P,Q$. 
Full-abstraction, 
i.e. $\forall P,Q: \tau(P) \leq_c \tau(Q)$ iff  $P \leq_c  Q$, 
only holds if the two calculi are more or less the same.
% DS: Kommt spaeter sowieso:
% The deeply investigated  translation $\tau_0$ mapping the $\pi$-calculus into $\CH$
% is defined and analysed in \cref{sec:translation}.
% We also exhibit further translations.

% \noindent\textbf{Source and Target Calculi.} %
\paragraph*{Source and Target Calculi.}
The well-known $\pi$-calculus \cite{milner-parrow-walker:92,milner-pi-calc:99,sangiorgi-walker:01} is a minimal model for \emph{mobile and concurrent processes}. 
Dataflow is expressed
%by communication between processes, \ie~
by passing messages between them via named channels, where messages are channel names. Processes and links between processes can be dynamically created and removed which makes processes mobile.
The interest in the $\pi$-calculus is not only due to the fact that it is used and extended for various  applications, like reasoning about cryptographic protocols \cite{abadi-gordon:97},
applications in molecular biology \cite{priami:95}, and distributed computing \cite{laneve-join:96,fournet-gonthier:2000}.
The $\pi$-calculus also permits the study of intrinsic principles and semantics of concurrency
% of concurrent  programming 
and the inherent nondeterministic behavior of mobile 
and communicating processes.
We investigate a variant of the $\pi$-calculus which is the synchronous $\pi$-calculus with replication, but without sums, matching operators, or recursion. 
To observe termination of a process, the calculus has a constant $\tStop$ which signals successful termination.
% In this calculus,  there is a context (\ie $(\cdot \PAR \cdot)$) which  is a parallel-convergence tester \cite{pucella-pananganden:01,plotkin:77}.
%\davidins{..}

\ignore{   das traegt nix zum Papier bei: 
One important reason why we use Concurrent Haskell as a target language,
and not sequential Haskell, is that for a correct translation, the target language must be able to perform a parallel convergence test  \cite{pucella-pananganden:01,plotkin:77}, since the $\pi$-calculus has one (\ie~by the context $[\cdot] \PAR [\cdot]$).
Neither sequential Haskell nor extensions by an erratic-choice have such a tester. 
 However,  Concurrent Haskell \cite{peyton-gordon-finne:96}  has a parallel-convergence testing context.
 }
% 
% since sequential Haskell and also  Haskell extended by erratic non-determinism are insufficient:
% a reason is that  the degree of expressiveness for concurrency can be assessed by so-called parallel-convergence testing  \cite{pucella-pananganden:01,plotkin:77}.
% In the $\pi$-calculus with $\tStop$ the context $(\cdot \PAR \cdot)$ is a parallel-convergence tester. Deterministic languages like Haskell do not have parallel-convergence testers.
% Core-Haskell with erratic choice (\ie, a construct that permits to choose between two expressions) appears to be more expressive, however, this and similar calculi 
% do  \emph{not} have a parallel-convergence testing context $P[\cdot,\cdot]$. 
% Informally, the reason is that if $P[v,\bot]$ as well as $P[\bot,v]$ converge for values $v$, then in such programming languages
% $P[e_1,e_2]$ has to converge independently of expressions $e_1,e_2$ (in particular, $P[\bot,\bot]$ must also converge).
% % , there are reduction sequences that choose  $e_1$ first,
% % as well as others that choose $e_2$ first, independent of $e_1, e_2$.\todo{Der Satz ist vom Aufbau nicht ganzh klar?}
% Note that there are stronger nondeterministic primitives, like  McCarthy's {\tt amb}-operator \cite{mccarthy63basis}, which can express 
% parallel-convergence testing contexts
% (see \eg~\cite{sabel-schmidt-schauss-MSCS:08}). 
% Concurrent Haskell \cite{peyton-gordon-finne:96,haskell-org:2019} is very powerful and has a parallel-convergence testing context, which also holds for $\CH$.
% }
The calculus $\CH$, 
% Wdh.
 a core language of Concurrent Haskell, 
is a process calculus where threads evaluate expressions from a lambda calculus extended by 
data constructors, case-expressions, recursive let-expressions, 
and Haskell's seq-operator. Also monadic operations (sequencing and creating threads) are available. The shared memory is modelled by MVars 
(mutable variables) which are one-place buffers that can be either filled or empty. The operation $\ttakeMVar$ tries to empty a filled MVar and blocks if the MVar is already empty. The operation $\tputMVar$ tries to fill an empty MVar and blocks if the MVar is already filled. 
The calculus $\CH$  is a variant (or a subcalculus) 
of the calculus $\CHF$ \cite{sabel-schmidt-schauss-PPDP:2011,sabel-schmidt-schauss-LICS:12} which 
extends Concurrent Haskell with futures.
% in order to increase the declarativeness of such languages.
A technical %% major 
  advantage of this approach is that we can reuse studies and results on the contextual semantics of $\CHF$ also for $\CH$.
% The advantage is that we can reuse studies and results on the contextual semantics of $\CHF$.  %%%, which can be transferred to $\CH$ via an embedding into $\CHF$. 

%   An open issue that we study and solve in this paper is whether the $\pi$-calculus can be embedded into $\CH$ under strong conditions.
%   A strong and well-behaved embedding would increase the knowledge 
%   of concurrency and the $\pi$-calculus on one hand,   and on the other, it leads to an executable implementation with the same semantics. 
  
 %\vspace*{2cm}
 
 %% \davidins{
  
% \noindent{\bfseries Details and Variations of the Translation.} 
\paragraph*{Details and Variations of the Translation.}  One main issue for a correct translation from $\pi$-processes to $\CH$-programs is to 
%correctly 
encode the 
  synchronous communication of the $\pi$-calculus. %% A considerable part of this paper will deal with this, and we will provide  different proposals.
  The problem is that the  MVars in $\CH$  have an asynchronous behavior (communication has to be implemented in two steps: 
  the sender puts the message into an MVar, which is later taken by the receiver).   %% in the next step
  To implement synchronous communication,  the weaker synchronisation property of MVars has to be exploited, 
    where we must be aware of the potential interferences of the executions of other translated communications on the same channel.
  %%first the sender 
 %% sends the message (by putting it into an MVar) and then the receiver has to send an acknowledgment to the sender (again using an MVar).
  % This communication must not be interfered with other senders and receivers on the same channel. 
  %
  The task of finding such translations is reminiscent of the channel-encoding used in \cite{peyton-gordon-finne:96}, but, however, there
  an asynchronous channel is implemented while we look for synchronous communication.
    
  We provide a translation $\tau_0$ which uses a private MVar per channel and per communication to ensure that no other process can interfere with the interaction.
  A similar idea was used in \cite{Honda:1991,boudol:1992} for keeping channel names private in a different scenario 
  %% \davidins{
  (see \cite{Glabbeek18,GlabbeekGLM19} for recent treatments of these encodings).  
  We %succeed in mathematically confirming 
  prove that the translation $\tau_0$ is correct. %% has strong and nice semantic properties. 
  Since we are also interested in simpler translations,    %(\eg~using private MVars is disadvantageous, since it generates a new MVar for each communication), 
  we looked for correct translations with a fixed and static number of MVars per channel in the $\pi$-calculus.
  Since this task is too complex and error-prone for hand-crafting, we automated it by implementing a tool to rule
  out incorrect translations\footnote{The tool and some output generated by the tool are available via \href{https://gitlab.com/davidsabel/refute-pi}{https://gitlab.com/davidsabel/refute-pi}.}. 
  Thereby we fix the MVars used for every channel:
  a single MVar for  
  exchanging the channel-name and perhaps several additional MVars of unit type to perform checks whether the message was sent or received (we call them check-MVars, 
 they behave like binary semaphores that are additionally blocking for signal-operations on an unlocked semaphore).
%   which have to be of unit type. 
  The outcomes of our  automated search are: a further correct translation that uses two check-MVars, where one is used as a mutex between all senders 
  or receivers on one channel, and
  further correct translations using three additional MVars where the filling and emptying operations for each MVar need not come from    %must
  the same sender or receiver. The experiments lead to the conjecture that there is no translation using only one check-MVar.
%   We also provide formal arguments for the correctness of some of the smallest translations with global MVars.

 % These results will also permit to draw the conclusion that the expressiveness of the $\pi$-calculus is completely available in Concurrent Haskell.

% \noindent {\bfseries Results.}
\paragraph*{Results.}
 Our novel result is convergence-equivalence and adequacy of the open translation $\tau$ (\cref{thm:may-must-equivalence,thm:adequate}), translating  the $\pi$-calculus into $\CH$. 
   The comparison of the $\pi$-calculus with a concurrent programming language 
   (here $\CH$) 
   using contextual semantics
   for may- and should-convergence in both calculi exhibits that the $\pi$-calculus is embeddable in $\CH$ 
       % keeping semantical properties of interest. 
         where we can prove that  the semantical properties of interest are kept.
   The adaptation of the adequacy and full abstraction notions 
   (\cref{def:adequate-and-fullyabstract}) for open processes is a helpful technical extension of our work in 
   \cite{schmidt-schauss-niehren-schwinghammer-sabel-ifip-tcs:08,schmidtschauss-sabel-niehren-schwing-tcs:15}.
    
 We further define a general formalism for the representation of translations with global names    %quite 
   and analyze different classes of such translations using an automated tool. 
   In particular, we show correctness of two particular translations in \cref{thm:tauinduced1-correct,thm:tauinduced2-correct}.   %of these 
 %
  %%% \davidins{
   The discovered correct translations look quite simple and their correctness seems to be quite intuitive. However, our experience is that searching for correct translations 
   is quite hard, since there are apparently correct (and simple) translations which were  wrong. Our automated tool helped us 
   to rule out wrong translations and to find potentially correct ones.
   %%} 
% \subsubsection*{Related work}  
\paragraph*{Discussion of Related Work on Characterizing Encodings.}
There are five criteria for valid translations resp.~encodings proposed and discussed in \cite{Gorla:10,GlabbeekGLM19},
which mainly restrict the translations w.r.t. language syntax and reduction semantics of the source and target language, called: compositionality, name invariance, 
operational correspondence, divergence reflection and success sensitiveness. Compositionality and name invariance restrict the syntactic form of the translated processes;
  operational correspondence means that the transitive closure of the reduction relation is transported by the translation, modulo the syntactic equivalence;
  and divergence reflection and success sensitiveness are conditions on the semantics.  
  
 In our approach, we define semantical congruence (and precongruence) relations on the source and target language. 
Thus the first two conditions  are not part of our notion of contextual equivalence, however, may be used as restrictions in showing non-encodability. 
We also omit the third condition and only use stronger variants of the fourth and fifth condition.
Convergence equivalence as a tool for finding out may-and should-convergence is our replacement of Gorla's divergence reflection and success sensitiveness.
We do not define an infinite reduction sequence as an error, which has as nice consequence that synchronization could be implemented by busy-wait techniques.

%%\manfredcomment{gleich weiter..}

%% \davidins{ ..} 
% \noindent {\bfseries Further Related Work.}  
\paragraph*{Further Related Work.}
Encodings of synchronous communication by asynchronous communication using a private name mechanism are given in  \cite{Honda:1991,boudol:1992} for (variants of the)
$\pi$-calculus. Our idea of the translation $\tau_0$ similarly
uses a private MVar to encode the channel based communication,
but our setting is different, since our
target language is Concurrent Haskell. 
Encodings between $\pi$-calculi with synchronous and with asynchronous communication
were, for instance, already considered in
\cite{Honda:1991,boudol:1992,Palamidessi03,Palamidessi:97} where encodability results are obtained for the $\pi$-calculus without sums \cite{Honda:1991,boudol:1992},
while in \cite{Palamidessi:97,Palamidessi03} the expressive power of synchronous and asynchronous communication in the $\pi$-calculus with \emph{mixed sums} was compared and 
non-encodability is a main result.
% however, with several restrictions on the encoding. For instance rule out encodings that do not encoding parallel composition into 
% parallel composition etc. 
% We do not use such restrictions, and 
% \ignore{
% However, we 
% consider translation between a message passing calculus and a shared memory model 
% and we focuse on the equivalence \wrt may- and should-convergence and the semantic notion of adequacy and full-abstractness. 
% }
Translations of the $\pi$-calculus into programming calculi and logical systems 
% and investigating the properties 
are given in \cite{banach-balazs-papadopoulos:95}, 
where a translation into a graph-rewriting calculus is given and  soundness and completeness \wrt the operational 
behavior is proved. 
The article \cite{yang-ramakrishnan-smolka:2004} shows a translation and a proof that the $\pi$-calculus is exactly operationally represented.
There are several works on session types which are related to the $\pi$-calculus, \eg,
 \citet{orchard-yoshida:16} studies encodings from a session calculus into PCF extended by concurrency and effects and also an embedding 
 in the other direction, mapping PCF extended by effects into a session calculus. The result is a (strong) operational correspondence 
 between the two calculi.
 In \cite{cano-et-al:2017} an embedding of a session $\pi$-calculus into ReactiveML is given and operational correspondence between the two languages is shown.
%  
%  \marginpar{\davidins{Referee:
% A related work, on encoding a pi-calculus in a functional programming language:
% Mauricio Cano, Jaime Arias, Jorge A. Pérez: Session-Based Concurrency, Reactively. FORTE 2017: 74-91
% }
% }
 %\davidins{}
% 
% 
%\david{Mehr oder weniger aus der Rebuttal-Antwort zusammen kopiert:}\\\davidins{
Encodings of CML-events in Concurrent Haskell using MVars are published in \cite{Russell01,Chaudhuri09}. 
This approach is more high-level than ours (since it considers events, while we focus the plain $\pi$-calculus).   In \cite{Chaudhuri09} correctness of a distributed protocol for 
selective-communication {w.r.t.}~an excerpt of CML is shown and a correct implementation of  the protocol in the $\pi$-calculus is given. The protocol is implemented in 
Concurrent-Haskell, but no correctness of this part is shown,
since \cite{Chaudhuri09} focuses to show that CML-events are implementable in languages with first-order message-passing which is different from our focus 
 (translating the $\pi$-calculus into $\CH$).
%}

%% \manfredins{The target calculus must have a powerful expressiveness, see \ref{remark:-encodable}}

%\subsubsection*{Outline}     
% \noindent {\bfseries Outline.} 
\paragraph*{Outline.}
%     The structure of the paper is as follows. 
    We introduce the source and target calculi in \cref{sec:pi-calc,section:chf}, 
    the translation using private names in \cref{sec:translation}, and
% 
%    defined
    %explained 
   %% \davidins{..}
%    and correctness is shown.
    in \cref{sect:global-translations} we treat translations with global names.
%     and present our tool.
%     of
%     (invalid) translations.   %  We also state correctness of   two translations with global names.}
%     In \cref{sec:properties-of-translation} 
%     We analyze properties of the translation and present our results
%     while the main part of the proof is given  in
%     \cref{sec:translation-convergencies}.
    We conclude and discuss future work in \cref{sec:concl}. 
%    To improve readability and 
   Due to space constraints  most proofs are in the technical report \cite{schmidt-schauss-sabel:frank-60:19}.
%    , and some are in a technical appendix. % \davidins{}
    
%     the embedding of $\CH$ into $\CHF$ together with transporting correctness results for program transformations is given in a technical appendix.
  
% \color{black}
%%% TEXEXPAND: END FILE ./intro-CHWS.tex
%%% TEXEXPAND: INCLUDED FILE MARKER ./calculus-pistop.tex
\section{\texorpdfstring{The $\pi$-Calculus with Stop}{The pi-Calculus with Stop}}\label{sec:pi-calc}
%%% TEXEXPAND: INCLUDED FILE MARKER ./figure-pistopcalc.tex
\begin{figure*}[t]
\begin{minipage}{.58\textwidth}
$\begin{array}{@{}r@{\,}c@{\,}l@{}}
P,Q \in \pistop &::=& \nu x.P \mid \outputxy{x}{y}.P \mid \inputxy{x}{y}.P \mid  {!P} \mid P \PAR Q \mid 0 \mid  \tStop
\\
 C\in\pistopC &::=& [\cdot] \mid \bar{x}(y).C \mid x(y).C \mid C \PAR P \mid P \PAR C \mid !C \mid \nu x.C
 \\
 \PCtxt \in \PCtxts_\pi  &::=& [\cdot] \mid \PCtxt  \PAR P \mid P \PAR \PCtxt  \mid \nu x.\PCtxt.
\end{array}$
 \caption{Syntax of processes $\pistop$, process contexts $\pistopC$ and reduction contexts $\PCtxts_\pi$
 where $x,y$ are names.
}\label{fig:syntax-pistop}
\vspace*{1mm}
\textbf{Interaction rule:}
(ia) $\inputxy{x}{y}.P \PAR \outputxy{x}{z}.Q \ia P[z/y] \PAR Q$

% \vspace*{1mm}

\begin{tabular}{@{}l@{~}p{.82\textwidth}@{}}\textbf{Closure:}&If $P\equiv \PCtxt[P'],P'\ia Q', \PCtxt[Q'] \equiv Q, \text{ and }\PCtxt \in \PCtxts_\pi$ then $P\sr Q$
\end{tabular}
\caption{Reduction rule and standard reduction in $\pistop$\label{fig:sr-pi}}
\end{minipage}\hfill\begin{minipage}{.39\textwidth}\centering
$\begin{array}{@{}r@{~}c@{~}l@{}}
P &\equiv& Q,~\text{if } P =_{\alpha} Q\\
P \PAR (Q \PAR R) &\equiv& (P \PAR Q) \PAR R\\
\nu x.(P \PAR Q) &\equiv& P \PAR \nu x.Q, \text{if }x \notin \FN(P)
\\
P \PAR 0 &\equiv& P \\
\nu x.0 &\equiv& 0 \\
\nu x.\tStop &\equiv& \tStop \\
\nu x. \nu y.P &\equiv& \nu y. \nu  x.P \\
P \PAR Q &\equiv& Q \PAR P \\
!P &\equiv& P \PAR !P
\end{array}
$
\caption{Structural congruence in $\pistop$\label{fig:congrpi}}
\end{minipage}
\end{figure*}

%%% TEXEXPAND: END FILE ./figure-pistopcalc.tex

We explain the synchronous $\pi$-calculus \cite{milner-parrow-walker:92,milner-pi-calc:99,sangiorgi-walker:01} without sums, with replication, 
extended with a constant~\tStop~\cite{sabel-schmidt-schauss-pistop:2015}, that signals successful termination.
% of the whole 
% process.
The $\pi$-calculus with~\tStop~and the $\pi$-calculus without~\tStop~but with so-called barbed convergences \cite{sangiorgi-walker:2001} are equivalent \wrt~contextual semantics
\cite{schmidt-schauss-sabel:frank-60:19}.
Thus, adding the constant $\tStop$ is not essential, 
% however,
but our arguments are
% treatment 
% % and the translations are 
% is 
easier to explain with $\tStop$.
%~included.
% for the  $\-$-calculus with \tStop. 

\begin{definition}[Calculus $\pistop$] 
 Let $\mathcal{N}$ be a countable set of 
   (channel) names. 
The syntax of \emph{processes} is shown in \cref{fig:syntax-pistop}.
%  $P,Q \in \pistop := \nu x.P \mid \outputxy{x}{y}.P \mid \inputxy{x}{y}.P \mid  {!P} \mid P \PAR Q \mid 0 \mid  \tStop$ where $x,y$ are names. 
%  \in \mathcal{N}$
% \end{definition}
% We briefly explain the language constructs. 
Name restriction  $\nu x.P$ restricts the scope of name $x$ to process $P$,
$P \PAR Q$ is the parallel composition of $P$ and $Q$,
the process $\outputxy{x}{y}.P$ waits on channel $x$ to output $y$ over channel $x$ and then becoming $P$, the process $\inputxy{x}{y}.P$ waits on channel $x$ to receive input, 
after receiving the input $z$, the process turns into $P[z/y]$ 
(where $P[z/y]$ is the substitution of all free occurrences of name $y$ by name $z$ in process $P$), 
the process $!P$ is the replication of process $P$, \ie{} it behaves like an infinite parallel combination of process $P$ with itself, the process $0$  
is the silent process, and \tStop{} is a process constant that signals success.
%termination.
We sometimes write $\inputxy{x}{y}$ instead of $\inputxy{x}{y}.0$ as well as $\outputxy{x}{y}$ instead of $\outputxy{x}{y}.0$.

Free names  $\FN(P)$, bound names $\BN(P)$, and $\alpha$-equivalence $=_\alpha$ in $\pistop$
are as usual in the $\pi$-calculus. A process $P$ is {\em closed} if $\FN(P)=\emptyset$. Let $\pistop^c$ be the closed processes in $\pistop$.
Structural congruence $\equiv$ is the least congruence  satisfying the laws shown in \cref{fig:congrpi}.
{\em Process contexts} $\pistopC$ and {\em reduction contexts} $\PCtxts_\pi$ are
defined in \cref{fig:syntax-pistop}.
Let $C[P]$ be the substitution of the hole $[\cdot]$ in $C$ by $P$. 
The reduction rule $\ia$ performs {\em interaction} and 
%of the $\pistop$-calculus 
% is the so-called {\em interaction} $(\ia)$ and 
standard reduction $\xrightarrow{sr}$ is its closure w.r.t. reduction contexts and $\equiv$ (see \cref{fig:sr-pi}).
Let  $\xrightarrow{sr,n}$ denote $n$ $\xrightarrow{sr}$-reductions and  $\xrightarrow{sr,*}$ denotes the reflexive-transitive closure of  
$\xrightarrow{sr}$.
A process $P \in \Pi_{\tStop}$ is {\em successful}, if  $P \equiv \PCtxt[\tStop]$ for some $\PCtxt\in\PCtxts_\pi$.
\end{definition}

\begin{remark}
We do not include ``new'' laws for structural congruences on the constant $\tStop$, like $\tStop \PAR \tStop$ equals $\tStop$, 
% or 
% even $\tStop \PAR P$ equals $\tStop$. 
% We did not want to include them,
since this would require to re-develop a lot of theory known from the $\pi$-calculus without $\tStop$.
% again for the calculus with $\tStop$. 
In our view, $\tStop$ is a mechanism 
for a notion of success  %an easy 
%ful termination, 
that can be easily replaced by other similar notions (e.g.~observing an open input or output as in barbed testing).
However, it is easy to prove those  equations \emph{on the semantic level}.
% e.g.~the contextual equivalence $\tStop \PAR P \sim_c \tStop$ holds 
(\ie~\wrt~$\sim_c$ as defined below in \cref{def:pi-simc}).
\end{remark}

As an example for a reduction sequence, consider sending name $y$ over channel $x$ and then sending name $u$ over channel $y$:
$
  %\nu x,y.
  (\inputxy{x}{z}.\outputxy{z}{u}.0 \PAR \outputxy{x}{y}.\inputxy{y}{x}.0)
\ia 
  (\outputxy{z}{u}.0[y/z] \PAR \inputxy{y}{x}.0)  \equiv 
    (\inputxy{y}{x}.0 \PAR \outputxy{y}{u}.0)  
\ia
  %  
%   \nu x,y.
  (0 \PAR 0) \equiv 
%   \nu x,y.0 \equiv
  0.
$
% \subsection{Process-Contexts and Reduction Contexts}

% \begin{definition} A {\em standard reduction} $\xrightarrow{sr}$ is the application of {$\ia$} within a reduction context (modulo structural congruence):
% $$
%     \inferrule{P\equiv \PCtxt[P'],P'\ia Q', \PCtxt[Q'] \equiv Q, \text{ and } \PCtxt \in \PCtxts}{P\sr Q}
% $$
% Let  $\xrightarrow{sr,n}$ denote $n$ $\xrightarrow{sr}$-reductions and  $\xrightarrow{sr,*}$ denotes the reflexive-transitive closure of  
% $\xrightarrow{sr}$.
% \end{definition}

% \begin{definition}
% Let  $\xrightarrow{sr,n}$ denote $n$ $\xrightarrow{sr}$-reductions and  $\xrightarrow{sr,*}$ denotes the reflexive-transitive closure of  
% $\xrightarrow{sr}$.
% 
% A process $P \in \Pi_{\tStop}$ is {\em successful}, if  $P \equiv \PCtxt[\tStop]$ for some $\PCtxt\in\PCtxts_\pi$.
% \end{definition}
 
For the semantics of processes, we observe whether
 standard reductions successfully terminate or not. 
Since reduction is nondeterministic,
we test whether there exists a successful reduction sequence (may-convergence), and we test whether all reduction possibilities are 
successful (should-convergence).
\begin{definition}
Let $P$ be a $\pistop$-process. We say process $P$ is {\em  may-convergent} and write $P{\maycon}$, if and only if there is a successful process $P'$ with
$P \xrightarrow{sr,*} P'$. We say $P$ is {\em  should-convergent} and write $P{\mustcon}$ if and only if for all $P'$: $P \xrightarrow{sr,*} P'$ implies $P'{\maycon}$.
If $P$ is not may-convergent, then $P$ we say is {\em must-divergent} (written $P{\Uparrow}$). If $P$ is not should-convergent, then we say it is {\em may-divergent} (written $P{\uparrow}$). 
\end{definition}

\begin{example}\label{ex:bsp1}
  The process $P:=\nu x,y.(\inputxy{x}{z}.0 \mid \outputxy{x}{y}.\tStop)$ is may-convergent ($P{\maycon}$) and should-convergent ($P{\mustcon}$), since
$P  \sr 0 \mid \tStop$ is the only $\xrightarrow{sr}$-sequence for  $P$.
The process  $P' :=  \nu x,y.(\inputxy{x}{z}.0 \mid \outputxy{x}{y}.0)$  is may- and must-divergent (\ie~$P'{\uparrow}$ and $P'{\Uparrow}$), since $P' \xrightarrow{sr} 0$ is the only $\xrightarrow{sr}$-sequence for $P'$.
\\
For 
% process
$P'' := \nu x,y.(\outputxy{x}{y}.0 \mid \inputxy{x}{z}.\tStop \mid \inputxy{x}{z}.0)$,
% shows that may-convergence and should-convergence are different.
%  there are two standard-reduction possibilities:
%  for $P''$:
we have    $P'' \sr \nu x,y.(\tStop \mid \inputxy{x}{z}.0)$
 and
    $P'' \sr \nu x,y.\inputxy{x}{z}.\tStop$.
 The first result is successful, and the second result is not successful. 
 Hence,  for $P''$ we have
 $P''{\maycon}$ and $P''{\maydiv}$.
%  \reportOnly{It is also  may-divergent, but not must-divergent.}
\end{example}
Should-convergence implies may-convergence, and  must-divergence implies may-divergence.

\begin{definition}\label{def:pi-simc}
  For $P, Q \in \Pi_{\tStop}$ and observation $\xi \in \{\maycon,\mustcon,\uparrow,\Uparrow\}$,  we define
   \mbox{$P\leq_\xi Q$} iff   $P\xi \implies Q\xi.$
   The $\xi$-contextual preorders $\leq_{c,\xi}$ and then
   $\xi$-contextual equivalences $\sim_{c,\xi}$  are defined as
    $$P\leq_{c,\xi} Q \text{ iff } \forall C \in \pistopC : C[P] \leq_\xi C[Q]
     \text{~~~and~~~}
    P\sim_{c,\xi}Q \text{ iff } P \leq_{c,\xi} Q \wedge Q \leq_{c,\xi} P
    $$
%  $$\begin{array}{r@{}c@{}llr@{}c@{}l}
%     $P\leq_{c,\xi} Q \text{ iff } \forall C \in \pistopC : C[P] \leq_\xi C[Q]$
%  and $P\sim_{c,\xi}Q \text{ iff } P \leq_{c,\xi} Q \wedge Q \leq_{c,\xi} P$.
%   $$\begin{array}{r@{}c@{}l@{}l@{}r@{}c@{}l}
%     P\leq_{c,\xi} Q &\text{ iff }& \forall C \in \pistopC : C[P] \leq_\xi C[Q]
%      &\text{~~~and~~~}
%     P\sim_{c,\xi}Q &\text{ iff }& P \leq_{c,\xi} Q \wedge Q \leq_{c,\xi} P
%     \end{array}$$
\emph{Contextual equivalence} of $\Pi_{\tStop}$-processes is defined as
$
    P\sim_{c}Q \text{ iff } P \sim_{c,\maycon} Q \wedge P \sim_{c,\mustcon} Q.
$
\end{definition}
\reportOnly{
Inspection of the reduction contexts and the $\ia$-reduction shows:
\begin{lemma}\label{lem:pi-closed-red}
Let $P$ be a $\Pi_{\tStop}$-process s.t.  $\FN(P) \subseteq \{x_1,\ldots,x_n\}$.
\begin{enumerate}
 \item   If $P \xrightarrow{sr,*} P'$ then
 $\nu x_1,\ldots,\nu x_n.P \xrightarrow{sr,*} \nu x_1,\ldots,x_n.P'$.
 \item
 If
 $\nu x_1,\ldots,\nu x_n.P \xrightarrow{sr,*} P'$
 then
 $P' \equiv \nu x_1,\ldots,x_n.P''$ and 
 $P \xrightarrow{sr,*} P''$.
 \end{enumerate}
 \end{lemma}
}
%% \davidins{
\begin{example}
For $Q := \nu x,y.(\outputxy{x}{y}.0 \mid \inputxy{x}{z}.\tStop \mid \inputxy{x}{z}.0)$,
we have $\tStop \sim_{c,\maycon} Q$ (which can be proved using the methods in \cite{sabel-schmidt-schauss-pistop:2015}), but 
$\tStop \not\sim_{c} Q$, since 
$\tStop\mustcon$ and $Q\maydiv$ and thus $\tStop \not\sim_{c,\mustcon} Q$. Note that ${\leq_{c,\mustcon}} = {\leq_c}$ holds in $\pistop$, since there is a context $C$ such that for all processes $P$: $C[P]\mustcon \iff P \maycon$ (see \cite{sabel-schmidt-schauss-pistop:2015}). For instance, the  equivalence $0 \sim_{c,\mustcon} Q$ does not hold, since $!0\mustdiv$ and $!Q\mustcon$ and thus the context $C=![\cdot]$ distinguishes $0$ and $Q$ w.r.t.~should-convergence.
\end{example}
%%   }
% \marginpar{\davidins{Referee: After Definition 2.5 I suggest to add an example in which  P $\sim_{c, Downarrow}$ Q holds, but $P \sim_{c, downarrow} Q$ does not hold.
% }}

Contextual preorder and equivalence are (pre)-congruences. 
Contextual preorder remains unchanged if observation is restricted to closing contexts:
\begin{lemma}
Let $\xi\in\{\downarrow,\uparrow,\Downarrow,\Uparrow\}$, $P,Q$ be 
$\Pi_{\tStop}$-processes.
Then $P \leq_{c,\xi} Q$ if, and only if
  $\forall C\in \pistopC$ such that $C[P]$ and $C[Q]$ are closed: $C[P] \leq_{\xi} C[Q]$.
\end{lemma}
\reportOnly{\begin{proof}
One direction is trivial.
For the other direction, we first consider $\xi ={\downarrow}$. Assume that 
  $\forall C\in \pistopC$ such that $C[P]$ and $C[Q]$ are closed: $C[P] \leq_{\maycon} C[Q]$.
  Let $C$ be an arbitrary context such that $C[P]\maycon$ and let 
  $\FN(C[P]\PAR C[Q]) = \{x_1,\ldots,x_n\}$. 
  \Cref{lem:pi-closed-red} 
    shows $\nu x_1,\ldots,x_n.C[P]\maycon$. The precondition shows $\nu x_1,\ldots,\nu x_n.C[Q]\maycon$.
  \Cref{lem:pi-closed-red}  shows $C[Q]\maycon$.
  
The case $\xi = {\Uparrow}$ holds, since ${\Uparrow} = {\neg{ \downarrow}}$: If
  $\forall C\in \pistopC$ such that $C[P]$ and $C[Q]$ are closed: $C[P] \leq_{\Uparrow} C[Q]$.
  then this shows $C[Q] \leq_{c,\downarrow} C[P]$ which is equivalent to $C[P] \leq_{c,\Uparrow} C[Q]$.
  
For $\xi = {\uparrow}$, assume that 
  $\forall C\in \pistopC$ such that $C[P]$ and $C[Q]$ are closed: $C[P] \leq_{\uparrow} C[Q]$.
  Let $C$ be an arbitrary context such that $C[P]\uparrow$ and let 
  $\FN(C[P]|C[Q]) = \{x_1,\ldots,x_n\}$. 
  \Cref{lem:pi-closed-red} and since adding or removing $\nu$-binders does not change success nor must-divergence we have
 $\nu x_1,\ldots,x_n.C[P]\uparrow$. Now the precondition shows $\nu x_1,\ldots,\nu x_n.C[Q]\uparrow$, and
  \Cref{lem:pi-closed-red} and since adding or removing $\nu$-binders does not change success nor must-divergence
  shows $C[Q] \uparrow$.
  Finally, the case $\xi={\Downarrow}$ follows by symmetry and since $P\uparrow \iff \neg(P \Downarrow)$.
  \end{proof}
}

\ignore{
\subsection{Replacing Structural Congruence by Structural Rewriting}
Since applying structural congruence is highly nondeterministic, 
%% not helpful if applied below input or output prefixes, 
 and in order to facilitate reasoning on standard reduction sequences, we define a more restrictive 
% (and directed) 
use of the congruence laws which is still nondeterministic, but only applies the laws on the surface of the process.
\begin{definition}\label{def:dsr}
% [$\xrightarrow{dsc}$- and $\xrightarrow{dsr}$-reduction]
Let $\xrightarrow{dsc}$ be the union of the following rules, 
where $\PCtxt\in\PCtxts_\pi$:
$$
\begin{array}{@{}l@{~}l@{~}l@{~}l@{}}
\text{(assocl)} & \PCtxt[P \PAR (Q \PAR R)] \to \PCtxt[(P \PAR Q) \PAR R]
&% \\[.5ex]
\text{(assocr)} &\PCtxt[(P \PAR Q) \PAR R] \to \PCtxt[P \PAR (Q \PAR R)]
\\[.5ex]
\multicolumn{2}{@{}l}{\text{(commute)}~\PCtxt[P \PAR Q] \to \PCtxt[Q \PAR P]}
&\multicolumn{2}{@{}l}{\text{(replunfold)}~\PCtxt[!P] \to \PCtxt[P \PAR !P]}
\\[.5ex]
\text{(nuup1)} & \PCtxt[(\nu z.P)\PAR Q] \to \PCtxt[\nu z.(P \PAR Q)],
&\text{(nuup2)} & \PCtxt[\nu x.\nu z.P] \to \PCtxt[\nu z.\nu x.P],
\\
&
\text{if $z$ does not occur free in $Q$}
&&
\text{if $x \not= z$}\\
\end{array}
$$
Let $\xrightarrow{dia}$ be the closure of $\xrightarrow{ia}$ by reduction contexts $\PCtxts_\pi$ and let $\xrightarrow{dsr}$ be defined as 
the composition $\xrightarrow{dsc,*}\cdot\xrightarrow{dia}\cdot\xrightarrow{dsc,*}$.
\end{definition}

We omit the proof of the following equivalences, but it can be constructed analogous to the proof given in 
\paperOnly{\cite{sabel-paper:OASIcs:2014:4582}}\reportOnly{\cite{sabel-report:OASIcs:2014:4582}} 
for barbed may- and should-testing.
The theorem  allows us to restrict standard reduction to  $\xrightarrow{dsr}$-reduction when reasoning on reduction sequences 
that witness may-convergence or may-divergence, resp.
\begin{theorem}\label{theo:pi-dsr} For all processes $P\in\pistop$ the following holds:\\
% \begin{enumerate}
%  \item 
1. $P\maycon$ iff $P \xrightarrow{dsr,*} \PCtxt[\tStop]$.
2. $P\maydiv$ iff $\exists P'$ such that $P \xrightarrow{dsr,*} P'$ and $P'\mustdiv$.
% \end{enumerate}
\end{theorem}
}
%%% TEXEXPAND: END FILE ./calculus-pistop.tex
%%% TEXEXPAND: INCLUDED FILE MARKER ./calculus-CH.tex
\section{\texorpdfstring{The Process Calculus $\CH$}{The Process Calculus CHF}}\label{section:chf}    
%%% TEXEXPAND: INCLUDED FILE MARKER ./figure-calculus-CH.tex
\begin{figure*}[t]
% \centering
$\begin{array}{@{}r@{\,}l@{\,}l@{\,}l@{}}
P \in \ProcCH &::= &(P_1\PAR P_2) 
                     \OR  \UTHREAD{e}  
                     \OR     \NEW x . P   
                     \OR     \MVAR{x}{e}  
                     \OR     \EMPTYMVAR{x}   
                     \OR      \SHARE{x}{e}      
\\[.5ex]
e\in\ExprCH &::= & x  \OR \lambda x.e \OR (e_1~e_2) \OR c~e_1\ldots e_{\arity(c)} \OR \tletrec~x_1{=}e_1,\ldots,x_n{=}e_n~\tin~e \OR m 
            \OR \tseq~e_1~e_2\\
            &&
            \OR\tcase_{T}\,e\,\tof\,(c_{T,1}\,x_1 \ldots x_{\arity(c_{T,1})} \casepf e_1) \ldots   (c_{T,|T|}\,x_1 \ldots x_{\arity(c_{T,|T|})} \casepf e_{|T|})
             \\[.5ex]
m \in \MExprCH &::=&  \treturn\,e \OR e\,\tbind\,e' \OR \tforkIO\,e 
                   \OR \ttakeMVar\,e \OR {\tt newMVar}\,e \OR \tputMVar\,e\,e'
\\[.5ex]                   
\typetau \in\TypCH &::= & {\tt IO}~\typetau ~|~  (T~\typetau_1 \ldots \typetau_n) ~|~ {\tt MVar}~{\typetau} ~|~ \typetau_1 \to \typetau_2
\end{array}
\\
\begin{array}{@{}r@{\,}l@{\,}l@{\,}l@{}}
\PCtxt \in \PCtxtsCH &::=& [\cdot] \OR \PCtxt \PAR P \OR P \PAR \PCtxt \OR \NEW x.\PCtxt 
\\
\MCtxt \in \MCtxtsCH          &::=& [\cdot] \OR \MCtxt \tbind e
\\
\end{array}\hfill
\begin{array}{@{}r@{\,}l@{\,}l@{\,}l@{}}
\ECtxt \in \ECtxtsCH &::=& [\cdot] \OR (\ECtxt\,e) \OR (\tseq\,\ECtxt\,e)\OR (\tcase\,\ECtxt\,\tof\,alts)
\\
\FCtxt \in \FCtxtsCH &::=& \ECtxt \OR (\ttakeMVar\,\ECtxt) \OR (\tputMVar\,\ECtxt\,e) 
\end{array}
$ 
\caption{Syntax of processes, expressions,  types, and context classes of $\CH$}\label{fig:syntax-expr}\label{fig:chf-contexts}

 \vspace*{3mm}

\begin{minipage}{\textwidth}
\centering $
\begin{array}{@{~}l@{~}l@{}}
\multicolumn{2}{@{}l}{\text{\bfseries Functional Evaluation:}}\\
% [0.5ex]
\text{(cpce)}& \UTHREAD{\MCtxt[\FCtxt[x]]}  \PAR x = e \reduce\UTHREAD{\MCtxt[\FCtxt[e]]}  \PAR x = e
\\
\text{(mkbinds)}&\UTHREAD{\MCtxt[\FCtxt[\tletrec\,x_1{=}e_1,\ldots,x_n{=}e_n\,\tin\,e]]} \reduce \NEW x_1\ldots x_n.(\UTHREAD{\MCtxt[\FCtxt[e]]} \PAR x_1 {=} e_1 \PAR\!\ldots\!\PAR x_n {=} e_n)
\\
\text{(beta)}&\UTHREAD{\MCtxt[\FCtxt[((\lambda x.e_1)~e_2)]]} \reduce  \UTHREAD{\MCtxt[\FCtxt[e_1[e_2/x]]]}
\\
\text{(case)}& \UTHREAD{\MCtxt[\FCtxt[\tcase_T~(c~e_1 \ldots e_n)~\tof~\ldots(c~y_1 \ldots y_n \casepf e)\ldots]]}
             \reduce   \UTHREAD{\MCtxt[\FCtxt[e[e_1/y_1,\ldots,e_n/y_n]]]}\\
%\\
\text{(seq)}& \UTHREAD{\MCtxt[\FCtxt[(\tseq~v~e)]]} \reduce \UTHREAD{\MCtxt[\FCtxt[e]]} 
% \text{ where $v$ is a functional value}
\text{~~~~where $v$ is a functional value}
\\[0.5ex]
\multicolumn{2}{@{}l}{\text{\bfseries Monadic Computations:}}\\[0.5ex]
\text{(lunit)} &\UTHREAD{\MCtxt[\treturn~e_1 \tbind~e_2]} \reduce \UTHREAD{\MCtxt[e_2~e_1]}
\\
% [1.5ex]
\text{(tmvar)} &\UTHREAD{\MCtxt[\ttakeMVar~x]} \PAR \MVAR{x}{e}\reduce \UTHREAD{\MCtxt[\treturn~e]} \PAR \EMPTYMVAR{x}
\\
% [1.5ex]
\text{(pmvar)} &\UTHREAD{\MCtxt[\tputMVar~x~e]} \PAR \EMPTYMVAR{x} \reduce \UTHREAD{\MCtxt[\treturn~()]} \PAR \MVAR{x}{e}
\\
% [1.5ex]
\text{(nmvar)} &\UTHREAD{\MCtxt[\tnewMVar~e]}  \reduce \NEW x.(\UTHREAD{\MCtxt[\treturn~x]} \PAR \MVAR{x}{e})
\\
% [1.5ex]
\text{(fork)}  &\UTHREAD{\MCtxt[\tforkIO~e]}  \reduce 
% \NEW z.
\UTHREAD{\MCtxt[\treturn~()]} \PAR
    \UTHREAD{e}\\
%     [1.1ex]
%     \text{where $z$ is fresh}\\
    %% klar, sonst wuerde es da stehen: and the new thread is not main}\\   %% a main thread}\\
% \text{(unIO)}  &\UTHREAD{\treturn~e}  \reduce y = e~\text{if the thread is not the main-thread}
% &\text{\davidins{(unIO)~removed}}
% \\
%  \multicolumn{2}{@{}l}{\text{{\bfseries Closure :}~If $P_1 \equiv \PCtxt[P_1']$, $P_2 \equiv \PCtxt[P_2']$, and $P_1' \reduce P_2'$ then $P_1 \reduce P_2$}.}
%  \\[.5ex]
%  \multicolumn{2}{@{}l}{\text{{\bfseries Capture avoidance:}~We assume capture avoiding reduction for}}\\
%  \multicolumn{2}{@{}l}{\hspace*{2.8cm}\text{all reductions.}}
 \multicolumn{2}{@{}l}{\text{{\bfseries Closure w.r.t. $\PCtxt$-contexts and $\equiv$:}~If $P_1 \equiv \PCtxt[P_1']$, $P_2 \equiv \PCtxt[P_2']$, and $P_1' \reduce P_2'$ then $P_1 \reduce P_2$}.}
 \\[.5ex]
 \multicolumn{2}{@{}l}{\text{{\bfseries Capture avoidance:}~We assume capture avoiding reduction for all reductions.}}
 \end{array}
 $
\caption{Standard reduction rules of $\CH$ (call-by-name-version)\label{fig:sr-rules}}
\end{minipage}
\end{figure*}
%%% TEXEXPAND: END FILE ./figure-calculus-CH.tex

The calculus $\CH$  (a variant of the language $\CHF$, \cite{sabel-schmidt-schauss-PPDP:2011,sabel-schmidt-schauss-LICS:12})
models a core language of Concurrent Haskell \cite{peyton-gordon-finne:96}.
We assume a partitioned set of {\em data constructors} $c$ where each family represents a type $T$. The data constructors of type $T$ are $c_{T,1},\ldots,c_{T,|T|}$ where each $c_{T,i}$ has 
an arity $\arity(c_{T,i}) \geq 0$.      
We assume that there is a type {\tt ()} with data constructor {\tt ()},  a type {\tt Bool} with  constructors {\tt True}, {\tt False}, a type {\tt List} with constructors {\tt Nil} and {\tt :} (written infix\reportOnly{ as in Haskell}), and a type {\tt Pair} with a constructor $(,)$ written $(a,b)$.

% The syntax of the calculus $\CH$ 
% has 
Processes $P \in \ProcCH$ in $\CH$ 
% on the top-layer which 
have expressions $e\in\ExprCH$  as subterms.
See 
% The syntax is shown in 
% Processes and expressions are defined by the grammars in 
\cref{fig:syntax-expr}
where $u,w,x,y,z$ are variables from an infinite set $\Var$.
% As in the $\pistop$-calculus, 
Processes are formed
by parallel composition  
``$\PAR$''. The $\NEW$-binder restricts the scope of a variable. A concurrent thread $\UTHREAD{e}$ evaluates $e$.
In a process there is (at most one) unique distinguished thread,
called  {\em main thread}, written as $\UMTHREAD{e}$. 
MVars are mutable variables which are empty or filled. A thread blocks if it wants to fill a filled MVar $\MVAR{x}{e}$ or empty 
an empty MVar $\EMPTYMVAR{x}$. Here $x$ is called the {\em name of the MVar}. 
Bindings $\SHARE{x}{e}$ 
model the global heap, 
of shared expressions, 
where 
$x$ is called a {\em binding variable}.
If  $x$ is a name of an MVar or a binding variable, then $x$ is called an {\em introduced variable}.%
\reportOnly{An introduced variable is visible to the whole process unless its scope is restricted by a $\NEW$-binder,
  \ie\ in $Q \PAR \NEW x.P$ the scope of introduced variable 
 $x$ is process $P$.}
\paperOnly{In $Q \PAR \NEW x.P$ the scope of
% introduced variable 
$x$ is $P$.}
A process is {\em well-formed}, if all introduced variables  
are pairwise distinct and there exists at most one main thread $\UMTHREAD{e}$. 

Expressions $\ExprCH$ consist of functional
and monadic expressions ${\MExprCH}$.
Functional expressions are variables, {\em abstractions} 
$\lambda x.e$, {\em applications} $(e_1~e_2)$,
{\em $\tseq$-expressions} $(\tseq~e_1~e_2)$,
{\em constructor applications} $(c~e_1~\ldots~e_{\arity(c)})$,
{\em $\tletrec$-expressions} $(\tletrec~x_1=e_1,\ldots,x_n=e_n~\tin~e)$, and
{\em \tcase$_T$-expressions} for every type $T$.
We abbreviate \tcase-expressions as $\tcase_T~e~\tof~alts$ where $alts$ are the {\em \tcase-alternatives} such that there is exactly one alternative $(c_{T,i}~x_1 \ldots x_{\arity(c_{T,i})} \casepf e_{i})$ for every constructor $c_{T,i}$ of type $T$,
where  $x_1,\ldots,x_{\arity(c_{T,i})}$ (occurring in the {\em pattern} $c_{T,i}~x_1 \ldots x_{\arity(c_{T,i})}$) 
are pairwise distinct variables that become bound with scope $e_{i}$.
We often omit the type index $T$ in $\tcase_T$.
% if it is clear from the context.  
In $\tletrec~x_1=e_1,\ldots,$ $x_n=e_n$ $\tin~e$ the variables $x_1,\ldots,x_n$ are pairwise distinct and the bindings $x_i = e_i$ are recursive, \ie\ the scope of
$x_i$ is $e_1,\ldots,e_n$ and $e$.
{\em Monadic operators} ${\tt newMVar}$, ${\ttakeMVar}$, and ${\tputMVar}$ are used to create, to empty and to fill MVars,
the ``bind''-operator {$\tbind$} implements
sequential composition of IO-operations, 
the {\tt forkIO}-operator performs thread creation, 
and {\tt return} lifts expressions into the monad.

 {\em Monadic values} are
% monadic expressions of the form
\mbox{$\tnewMVar\,e$},
\mbox{$\ttakeMVar\,e$},
\mbox{$\tputMVar\,e_1\,e_2$},
\mbox{$\treturn\,e$}, 
\mbox{$e_1\,\tbind\,e_2$},
or 
\mbox{$\tforkIO\,e$}.
{\em Functional values} are abstractions and constructor applications.
A {\em value} is a functional or a monadic value.

Abstractions, $\tletrec$-ex\-pres\-sions, $\tcase$-alternatives, 
and   $\NEW x. P$ introduce variable binders.
This induces bound and free variables (dentoted by $\FV(\cdot)$), $\alpha$-renaming, and $\alpha$-equivalence $=_{\alpha}$.
If $\FV(P)= \emptyset$, then we call process $P$ \emph{closed}.
We assume the {\em distinct variable convention}: free variables are distinct from bound variables, and bound variables are pairwise distinct.  
We assume that $\alpha$-renaming is applied to obey this convention.
Structural congruence $\equiv$ of $\CH$-processes is the least congruence satisfying the laws 
$P_1 \PAR P_2 \equiv P_2 \PAR P_1$,
$(P_1 \PAR P_2) \PAR P_3 \equiv  P_1 \PAR (P_2 \PAR P_3)$,
$\NEW x_1.\NEW x_2.P  \equiv  \NEW x_2.\NEW x_1.P$,
$P_1  \equiv  P_2  \text{ if $P_1 =_{\alpha} P_2$}$,
and
$(\NEW x.P_1)\PAR P_2      \equiv     \NEW x.(P_1 \PAR P_2) \text{ if $x \not\in \FV(P_2)$}$.
%  
% \cref{fig:sc}. 
% It allows to treat parallel composition as an associative-commutative operator, to shift $\nu$-bindefrs to the top-level of processes and to $\alpha$-rename processes.

We assume expressions and processes to be well-typed \wrt~a monomorphic type system: the typing rules are standard (they can be found in \cite{schmidt-schauss-sabel:frank-60:19}). The syntax of types is in \cref{fig:syntax-expr} where
$({\tt IO}~\typetau)$ is the type of  a monadic action with return type $\typetau$, 
$({\tt MVar}~\typetau)$ is the type of an {\tt MVar} with content type $\typetau$, 
and $\typetau_1 \to \typetau_2$ is a function type.
We treat constructors like overloaded constants to use them in a polymorphic way. 

We introduce a call-by-name small-step reduction for $\CH$.
This operational semantics can be shown to be equivalent to a 
call-by-need semantics (see 
% 
% The proof is a copy of the proof given in 
\cite{sabel-schmidt-schauss-PPDP:2011}
for the calculus $\CHF$). However, the equivalence of the reduction strategies is not important for this
 paper. 
 That is why we do not include it. 

In $\CH$, a context is a process or an expression with a (typed) hole $[\cdot]$.   
We introduce several classes of contexts in \cref{fig:chf-contexts}.
They are used by the reduction rules.
\reportOnly{
On the process level there are {\em process contexts} $\PCtxtsCH$,
on expressions first {\em monadic contexts} $\MCtxtsCH$ are used to find the next to-be-evaluated monadic action in a sequence of actions.
For the evaluation of (purely functional) expressions, usual (call-by-name) {\em expression evaluation contexts} $\ECtxt \in \ECtxtsCH$ are used, 
and to enforce the evaluation of the (first) argument of the monadic operators $\ttakeMVar$ and $\tputMVar$
the class of {\em forcing contexts} $\FCtxt \in \FCtxtsCH$ is used.
}
\begin{definition}
The standard reduction $\reduce$ is defined by the rules and the closure in \cref{fig:sr-rules}.
It is only permitted
% We permit standard reduction only 
for well-formed processes which are not successful.
\end{definition}
Functional evaluation includes 
$\beta$-reduction (beta), copying shared bindings into needed positions (cpce),  evaluating $\tcase$-
and $\tseq$-expressions (case) and (seq), and moving
$\tletrec$-bindings into the global bindings (mkbinds). 
For monadic computations, rule (lunit) 
implements the monadic evaluation.
Rules (nmvar), (tmvar), and (pmvar) handle the MVar creation and access. 
A \ttakeMVar-operation can only be performed on a filled MVar, and a \tputMVar-operation needs an empty MVar.
Rule (fork) spawns a new  thread.
A concurrent thread of the form 
$\UTHREAD{\treturn~e}$
is terminated (where $e$ is of type {\tt ()}).

\begin{example}
The process
%% $\MTHREAD{\!\!}\tnewMVar\,\texttt{()} \tbind (\lambda y. \tforkIO\,(\ttakeMVar\,y)) \tbind\lambda \_.\tputMVar\,y\,\texttt{()}$
$\MTHREAD{\!\!}\tnewMVar\,\texttt{\rm ()} \tbind (\lambda y. \tforkIO\,(\ttakeMVar\,y)) \tbind\lambda \_.\tputMVar\,y\,\texttt{\rm ()}$
creates a filled MVar, that is emptied by a spawned thread, and then again filled by the main thread.
%%%\marginpar{\david{hier die Reduktion rausgenommen}}
% In \cref{fig:example-eval-CH} we show a standard reduction sequence of this process.
% \input{figure-example-evaluation-CH}
\end{example}

We say that a $\CH$-process $P$ is {\em successful}  if $P \equiv \NEW x_1 \ldots  x_n. (\UMTHREAD{\treturn~e} \PAR P')$ and if $P$ is well-formed.
% \reportOnly{These are the desired results of standard reduction sequences.} 
% The  successful processes 
This  captures Haskell's behavior that termination of the main-thread  terminates all threads.

\begin{definition}\label{def:CH-convergence}
 Let $P$ be a $\CH$-process. Then $P$ {\em may-converges} (denoted as $P\maycon$), iff 
$P$ is well-formed and $\exists P': P \reduce[,*] P' $ such that $P'$ is  \text{ successful}.
If $P\maycon$ does not hold, then $P$ {\em must-diverges} and we write $P\mustdiv$.
Process $P$ {\em should-converges} (written as $P\mustcon$), iff 
 $P$ is well-formed and $\forall P': P \reduce[,*] P' \implies P'\maycon$.
If $P$ is not should-convergent, then we say $P$ {\em may-diverges} written as $P\maydiv$.
\end{definition}
\reportOnly{Note that a process $P$ is may-divergent iff there is a finite reduction sequence $P \reduce[,*] P'$ such that $P'\mustdiv$.}
\ignore{\Cref{def:CH-convergence} implies that  non-well-formed processes are always must-divergent, since they are irreducible and never successful.
}

\begin{definition}\label{def:simc}
Contextual approximation $\leq_c$ and equivalence $\sim_c$ on $\CH$-processes are defined as
${\leq_c} := {{\leq_{c,\maycon}} \cap {\leq_{c,\mustcon}}}$ and ${\sim_c} := {{\leq_{c}} \cap {\geq_{c}}}$
where 
% \begin{itemize}
% \item  
$P_1 \leq_{c,\maycon} P_2 \text{ iff } \forall \PCtxt\in\PCtxtsCH: \PCtxt[P_1]\maycon \implies \PCtxt[P_2]\maycon$
% \item
and 
$P_1 \leq_{c,\mustcon} P_2 \text{ iff } \forall \PCtxt\in\PCtxtsCH: \PCtxt[P_1]\mustcon \implies \PCtxt[P_2]\mustcon$.
% \end{itemize}
For $\CH$-\emph{expressions}, let $e_1 \leq_{c} e_2$ iff for all process-contexts $C$ with a hole at expression position: $C[e_1] \leq_c C[e_2]$
and $e_1 \sim_c e_2$ iff $e_1 \leq_c e_2 \wedge e_2 \leq_c e_1$.

\reportOnly{
A program transformation 
$\eta \subseteq (\ProcCH \times \ProcCH)$
is {\em correct} iff $P_1\ \eta\ P_2 \Longrightarrow P_1 \sim_c P_2$.
}
\end{definition}
%%%  \davidins{
As an example, we consider the processes 
$$
\begin{array}{lcl}
P_1 &:=&\nu m.(\MTHREAD{}{\ttakeMVar~m} \PAR 
  \THREAD{}{\ttakeMVar~m} \PAR \MVAR{m}{()})
\\
P_2 &:=& \MTHREAD{}{\treturn~()}
\\
P_3  &:=& \MTHREAD{}{\tletrec~x=x~\tin~x}
\end{array}
$$
Process $P_1$ is may-convergent and may-divergent (and thus not should-convergent), since either the main-thread succeeds in emptying the MVar $m$, or (if the other threads empties the MVar $m$) the main-thread is blocked forever. The process $P_2$ is sucessful. The process $P_3$ is must-divergent. 
  The equivalence $P_1 \sim_{c,\maycon} P_2$ 
holds, but $P_1\not\sim_{c} P_2$, since 
$P_2$ is should-convergent and thus $P_1 \not\sim_{c,\mustcon} P_2$.  As a further example, it is easy to verify that $P_1 \sim_{c,\mustcon} P_3$ holds, since both processes are not should-convergent and a surrounding context cannot change this. However, $P_1 \not\sim_{c,\maycon} P_3$, since
$P_3\mustdiv$.
%%%  }
% \marginpar{\davidins{Referee: After Definition 3.4 I suggest to add an example in which  P $<=_{c, Downarrow}$ Q holds, but P $<=_{c, downarrow}$ Q does not hold.}}

Contextual approximation  and equivalence are (pre)-congruences.
The following equivalence will help to  prove properties of our translation.  
\begin{lemma}\label{lem:conv-open-equiv-closed}
The relations in \cref{def:simc} are unchanged, if we add closedness: for $\xi \in \{\maycon,\mustcon\}$, let
$P_1 \leq_{c,\xi} P_2$ iff
$\forall \PCtxt\in\PCtxtsCH$ such that
$\PCtxt[P_1],\PCtxt[P_2]$ are closed:
$\PCtxt[P_1]\xi {\implies} \PCtxt[P_2]\xi$.

% \noindent $P_1 \leq_{c,\mustcon} P_2 \text{ iff } \forall \PCtxt\in\PCtxtsCH \text{ such that } \PCtxt[P_1],\PCtxt[P_2] \text{ are closed}: \PCtxt[P_1]\mustcon 
% \Rightarrow
% \implies
% \PCtxt[P_2]\mustcon$
% \end{itemize}   
\end{lemma}
\reportOnly{
\begin{proof}
 One direction is obvious. 
 For the other direction, let 
 $\xi \in \{\maycon,\mustcon\}$ and assume that for all $\PCtxt$ such that $\PCtxt[P_1],\PCtxt[P_2]$ are closed: $\PCtxt[P_1]\xi \implies \PCtxt[P_2]\xi$.
 Assume that $\PCtxt[P_1]\xi$ and $\FV(\PCtxt[P_1]) \cup \FV(\PCtxt[P_2]) = \{x_1,\ldots,x_n\}$.
 Since reductions are applicable with or without $\nu$-binders on the top, we have $\nu x_1,\ldots,x_n.\PCtxt[P_1]\xi$ and by the precondition
 $\nu x_1,\ldots,x_n.\PCtxt[P_2]\xi$, since 
 $\nu x_1,\ldots,x_n.\PCtxt[\cdot]$ is a $\PCtxts$-context.
From $\nu x_1,\ldots,x_n.\PCtxt[P_2]\xi$ also
$\PCtxt[P_2]\xi$ follows, since reductions are applicable with or without $\nu$-binders on the top.
This shows $P_1 \leq_{c,\xi} P_2$.\qed
\end{proof}
}
\reportOnly{\begin{proposition}\label{prop:equiv-is-sim}
Let $P_1, P_2$ be well-formed and $P_1 \equiv P_2$. Then $P_1 \sim_c P_2$.
\end{proposition}}
\ignore{
The following proposition holds, since correctness of transformations
can be transferred 
from the calculus $\CHF$ \cite{sabel-schmidt-schauss-PPDP:2011}
by an embedding $\iota:\CH \to \CHF$,
 such that $\iota(P) \sim_{c,\CHF} \iota(P')$ implies $P \sim_{c,\CH} P'$. 
 See \cref{app-ch-chf} for the full proof.
%  }
%  
}
\ignore{
\begin{proposition}\label{prop:det-sr-correct}
The rules \textit{(lunit), (nmvar), (fork)}, \textit{(cpce)}, \mbox{(mkbinds)}, \textit{(beta), (case)},  \textit{(seq)},
\textit{(gc)},
\textit{(dtmvar)}, and \textit{(dpmvar)}  
are correct program transformations.
% Transformations \textit{(pmvar)} and \textit{(tmvar)} are not correct.
% as program transformations.
\end{proposition} 
}
%%Manfred:   so nicht geraucht, cpce = gc
% \begin{definition}
%  Let (gc), (gcp) and (dtmvar) be program transformations defined as follows.
% $$\begin{array}{@{}l@{\,}l@{}}
%   \text{(gc)} & \nu x_1,\dots, x_n.
%     (P \PAR \texttt{\tt Comp}(x_1) \PAR \dots \PAR \texttt{\tt Comp}(x_n)) \to P, 
%              \mbox{if}~ \forall i \in \{1,\ldots,n\}: \texttt{\tt Comp}(x_i) \mbox{ is}\\
%              &\mbox{a binding } x_i = e_i, \mbox{ an MVar } \MVAR{x_i}{e_i} \mbox{ or } \EMPTYMVAR{x_i}, \mbox{and } x_1,\dots,x_n \notin \FV(P) 
% \\[.5ex]
% \text{(gcp)}& \Ctxt[x] \PAR x = e \to \Ctxt[e] \PAR x = e 
% \\[.5ex]
% \text{(dtmvar)}& \nu x.\PCtxt[y \Leftarrow \MCtxt[\texttt{\tt takeMVar x}] \PAR \MVAR{x}{e}] 
%     \to \nu x.\PCtxt[y \Leftarrow \MCtxt[\treturn~e] \PAR \EMPTYMVAR{x}],
%     \mbox{if }\forall \PCtxt'{\in}\PCtxts\\
%   &
% %   \mbox{if } \forall \PCtxt' \in \PCtxts 
%   \mbox{and all standard reductions starting with } 
%   \PCtxt'[\nu x.(\PCtxt[y \Leftarrow \MCtxt[\ttakeMVar~x] \PAR \MVAR{x}{e}])] \\&\mbox{the first execution of } (\ttakeMVar~x) \mbox{ is in thread }  y .
% \end{array}
% $$
% \end{definition}
% 
% 
% \begin{proposition}[\cite{frank44}]\label{prop:gc-etal-correct}
% The transformations (gc), (gcp) and (dtmvar) are correct.
% \end{proposition}

\ignore{ Wird nicht mehr gebraucht, soweit ich sehe \ldots.
\begin{theorem}\label{thm:standardisierung-CHF-all}
Let $P,Q$ be $\CH$-processes and $P \xrightarrow{*} Q$ be a reduction sequence consisting of $\reduce[]$-steps and correct program transformations.
Then $Q\maycon$ implies $P\maycon$, and $Q\maydiv$ implies $P\maydiv$.
\end{theorem}
% \begin{proof}
% Induction on the number of reductions of $P \xrightarrow{*} Q$.\\
% If $P \xrightarrow{*} Q$ only consists of correct transformations, then $P \sim_c Q$, and hence $P\maycon$.
% If $P \xrightarrow{a} P'  \xrightarrow{*} Q$ where $P \xrightarrow{a} P'$ is a standard reduction,  % transformation, \ie~a \tputMVar or \ttakeMVar-reduction. 
% then by induction hypothesis, $P'\maycon$, and  there is a standard reduction $P' \xrightarrow{sr,*} Q'$  to a successful process $Q'$. The sequence
% $P \xrightarrow{a} P'\xrightarrow{sr,*} Q'$ is the desired standard reduction.
% \end{proof}
%\davidins{..}
\begin{proof}
First we show $Q\maycon$ implies $P\maycon$ by induction on the number of reductions and transformation steps of $P \xrightarrow{*} Q$. If there are no steps,
then the claim holds (since $P = Q$).
If $P \to P' \xrightarrow{*} Q$, then the induction hypothesis shows that $P'\maycon$, \ie{} there exists a successful process $Q'$ such that $P' \xrightarrow{sr,*} Q'$.
We distinguish two cases: (i) If $P \to P'$ by a correct transformation, then $P \sim_c P'$ and thus the definition of $\sim_c$ implies $P\maycon$.
 (ii) If $P \xrightarrow{sr} P'$, then $P \xrightarrow{sr,*} Q'$ and thus $P'\maycon$.
 The proof of the second case is completely analogous.
\end{proof}
}

\section{\texorpdfstring{The Translation $\tau_0$ with Private MVars}{The Translation tau0 with Private MVars}}\label{sec:translation}
%%% TEXEXPAND: INCLUDED FILE MARKER ./figure-translation-tau.tex
\begin{figure*}[t]
$\begin{array}{@{}l@{\,}l@{\,}l@{}}
\tau_0(P)  &= &  \UMTHREAD{\bfdo~\{}
                 \begin{array}[t]{@{}l@{}}
                  \istop \leftarrow \tnewMVar~();  
%                   \\
                  \tforkIO~\tau(P);  
%                   \\
                  \tputMVar~\istop~()\}
                 \end{array}
%             \\\\[-2.8ex]
\\
    \tau(\outputxy{x}{y}.P)  & =& \bfdo~\{
    \begin{array}[t]{@{}l@{}}
        checkx \leftarrow \tnewMVar~();
%         \\
        \tputMVar~(\mathit{unchan}~x)~ (y,checkx);
%         \\
        \tputMVar~checkx~();
%         \\
        \tau(P)\}
    \end{array}
%             \\\\[-2.8ex]
\\
\tau(\inputxy{x}{y}.P) & =& \bfdo~\{
    \begin{array}[t]{@{}l@{}}
        (y,checkx) \leftarrow \ttakeMVar~(\mathit{unchan}~x);
%         \\
        \ttakeMVar~checkx;
%         \\
        \tau(P)
         \}
        \}
    \end{array}
%             \\\\[-2.8ex]
% % % %  
\\
\tau(P \PAR Q) &  =& \bfdo~\{\tforkIO~\tau(Q);\tau(P)\}
%             \\\\[-2.8ex]
\\
\tau(\nu x.P) &  =&  \bfdo\,\{
        \begin{array}[t]{@{}l@{}}
        \mathit{chanx} \leftarrow \tnewEmptyMVar;
%         \\
        \tletr~ x = \tChan\,\mathit{chanx}~\tin~\tau(P)\}\\
     \end{array}
%             \\\\[-2.8ex]
\\
\tau(0) &  = & {\treturn~()} 
\\
% \\\\[-2.8ex]
      \tau(\tStop) &  = & {\ttakeMVar~\istop}  
\\
% \\\\[-2.8ex]
         \tau(!P) & = & {\tletr~ f = \bfdo~\{\tforkIO~\tau(P); f\}~\tin~ f}  
\end{array}$
\caption{Translations $\tau_0$ and $\tau$\label{fig:def-tau}}
\end{figure*}
%%% TEXEXPAND: END FILE ./figure-translation-tau.tex

%%%%%%%%%%%%%%%%%%%%%%%%%%%%
We present a translation $\tau_0$ that encodes $\pistop$-processes as $\CH$-programs. 
It establishes correct synchronous communication by using a  private MVar,
which is created by the sender and its name is 
sent to the receiver. 
The receiver uses it to acknowledge that the message was received. 
Since only the sender and the receiver know this  MVar, no other thread can interfere the communication.The approach has
similarities
with Boudol's translation \cite{boudol:1992} from the $\pi$-calculus into an asynchronous one, where 
% which  was an early example of using 
a private channel name of the $\pi$-calculus was used to guarantee safe communication between sender and receiver.
 
% within a translation.
% Later in \cref{sect:global-translations}, we discuss translations and investigations 
% on those translations which do not use private names.
%%%%%%%%%%%%%%%%%%%%%%%%%%%%
%\davidins{}
% 
% %%%%%%%%%%%%%%%%%%%%%%%%%%%%
% \loeschen{{
% Our results in the last section justify the introduction of private names for our  translation $\tau_0$. 
% Note that the translation of the pi-calculus into an asynchronous one using private names  \cite{boudol:1992} was an early example of using private names within 
% a translation.} 
% }
%%%%%%%%%%%%%%%%%%%%%%%%%%%%
% 
% We define a translation $\tau_0$ mapping $\Pi_{\tStop}$-processes to $\CH$-processes. 

For translating $\pi$-calculus channels into $\CH$, we use a recursive data type $\tChannel$ (with constructor $\tChan$), which is defined in 
Haskell-syntax as 
 $$ \texttt{data}~\tChannel  = \tChan~(\tMVar~(\tChannel,(\tMVar~ ())))$$  %   
%  It is a recursive data type and can be initialized with a $\bot$-expression. 
We abbreviate  
$(\tcase_\texttt{Chan}~e~\tof~(\tChan~m~\casepf~m))$ as $(\mathit{unchan}~e)$.
%%\color{black}

We use $a~\tbindthen~b$ for $a \tbind  (\lambda \_\,.\, b)$ and also use Haskell's do-notation with the following meaning:
$$
\begin{array}{@{}l@{~}l}
\bfdo~\{x \leftarrow e_1; e_2\} &=  
        e_1~\tbind~\lambda x. (\bfdo~\{e_2\})
\\
\bfdo~\{(x,y) \leftarrow e_1; e_2\} &=        e_1 \tbind~\lambda z. \tcase_{\tt Pair}~z~\tof~(x,y)\casepf (\bfdo~\{e_2\})
\end{array}\qquad
\begin{array}{l@{~}l}
\bfdo~\{e_1; e_2\}& = e_1 \tbindthen~ (\bfdo~\{e_2\}) 
\\
\bfdo~\{e\}  &=  e 
% \\
% \bfdo~\{\textbf{let}~x=e_1;e_2\} &=  
%         \tletrec~x=e_1~\tin~\bfdo~\{e_2\}
\end{array}$$
% \david{mal das let hier als zucker eingeführt, damit es später kürzer wird}.
% $$
% \begin{array}{ll}
%  \bfdo~\{x \leftarrow e_1; e_2\} 
%     &=  
%         e_1 \tbind~\lambda x. (\bfdo~\{e_2\})
% \\
%  \bfdo~\{(x,y) \leftarrow e_1; e_2\} 
%     &=  
%         e_1 \tbind~\lambda z. \tcase_{\tt Pair}~z~\tof~(x,y)\casepf (\bfdo~\{e_2\})
% \\
%  \bfdo~\{e_1; e_2\}  
%     &=  
%         e_1 \tbindthen~ (\bfdo~\{e_2\}), 
% \\
% %%   Ms, ok
% % \loeschen{\bfdo~\{\tletrec~x=e_1~\tin~e_2\}}
% %     &= 
% %        \loeschen{ \tletrec~x=e_1~\tin~ (\bfdo~\{e_2\})}
% %  
% % \\
% \bfdo~\{e\}  
%     &=  e 
% \end{array}
% $$
% 
%%\davidins{}
As a further abbreviation, we write 
$y\leftarrow \tnewEmptyMVar$ inside a \bfdo-block to abbreviate the sequence 
$y \leftarrow \tnewMVar~\bot; \ttakeMVar~y$,
where $\bot$ is a must-divergent expression.
Our translation uses one 
% from the $\pistop$-calculus is done by using one 
MVar per channel 
which contains a pair consisting of the (translated) name of the channel and a further
MVar used for the synchronization, which is private, i.e.~only the sender and the receiver know it. Privacy is established by the sender: it creates a new MVar for every send operation.
% ceating two MVars per channel, 
% the first one contains the (translated name) of the channel and the second
% is for the synchronization. 
% % \ignore{\begin{center}
% % \begin{tikzpicture}[scale=.8,yscale=.7]
% %  \node[komBox](box){};
% %  \node[above =0em of box.center] (send) {\textbf{send}};
% %  \node[below = 0.5em of box.center] (check) {\textbf{check}};
% %  \node[above = 0.4em of box] (channel) {Chan $x$};
% %  \node[left = 2.4em of channel] (sender) {sender};
% %  \node[below = 1.8em of sender] (versendetLinks) {$y$};
% %  \draw[thick,black] (sender) -- (versendetLinks);
% %  \draw[pfeil] (versendetLinks) to (versendetLinks-|box.west);
% %  \node[below = 0.8em of versendetLinks,draw=black,fill =black] (empfaengtLinks) {};
% %  \draw[pfeil] (empfaengtLinks-|box.west) to (empfaengtLinks);
% %  \node[below = 1.8em of empfaengtLinks] (AusgangL) {};
% %  \draw[pfeil] (versendetLinks) to (AusgangL);
% %  \node[right = 2em of channel] (receiver) {receiver};
% %  \node[below = 1.8em of receiver,draw=black,fill =black] (erhaeltRechts) {};
% %  \draw[pfeil](erhaeltRechts-|box.east) to (erhaeltRechts);
% %  \node[below = 1.5em of erhaeltRechts,draw=black,circle,fill =black,inner sep = 1pt] (versendetRechts){};
% %  \draw[pfeil](versendetRechts) to (versendetRechts-|box.east) ;
% %  \node[below = 1.2em of versendetRechts] (Ausgang){$y$};
% %  \draw[pfeil](receiver) to (Ausgang);
% % \end{tikzpicture}
% % \end{center}
% % }
Message $y$ is sent over channel $x$ by sending a pair $(y,check)$ where $check$ is an MVar containing ().
The receiver waits  %(black square)
 for a message $(y,check)$ by the sender.
After sending the message, the sender waits until \textit{check} is emptied, and the receiver acknowledges 
by emptying the MVar  \textit{check}\footnote{A variant of the translation would be to change the roles for the acknowledgement such that 
an empty MVar is created, which is filled by the receiver and emptied by the sender. 
The reasoning on the correctness of the translation is very similar to the one presented here.}

\begin{definition}\label{definition:tau}  
We define the  translation $\tau_0$ and its inner translation $\tau$ 
from the $\pistop$-calculus into the $\CH$-calculus  %% $\CHF$-calculus 
    in \cref{fig:def-tau}.
% as follows.
For contexts, the translations are the same where the context hole is treated like a constant and translated as $\tau([\cdot])=[\cdot]$.
 \end{definition}
The translation $\tau_0$ generates a main-thread and an MVar $\istop$.  
The main thread is then waiting for the MVar $\istop$ to be emptied. 
The inner translation $\tau$ translates the constructs and constants
of the $\Pi_{\tStop}$-calculus into $\CH$-expressions.    
Except for the main-thread (and using keyword \texttt{let} instead of $\tletrec$), the translation $\tau_0$ 
generates a valid Concurrent Haskell-program, \ie~if we write $\tau_0(P) = \UMTHREAD{e}$ 
as $\texttt{main} = e$ , we can execute the translation in the Haskell-interpreter. 
%\davidins{\ldots}
% Here we assume that the \texttt{future}-primitive is implemented as suggested in 
% \cite{sabel-schmidt-schauss-PPDP:2011}. An alternative is to replace all expressions 
% (\texttt{future}~$e$) by (\texttt{forkIO}~$e$) and changing the type of these expressions 
% from $\texttt{IO}~\typetau$ to $\texttt{IO}~()$. This is possible, since the translated program
% never uses the resulting futures generated by $\texttt{future}$. 
% 
\ignore{The translation $\sigma_0$, which will be defined below, does not generate a Haskell program, since it maps to process components.}

\begin{example}
%%  \marginpar{\david{Das Beispiel  stark gekürzt}}
% \color{blue}
We consider the $\pistop$-process 
$P:=\nu x,y_1,y_2,z.(\inputxy{x}{y_1}.0\PAR\inputxy{x}{y_2}.\tStop\PAR\outputxy{x}{z}.0)$
which is may-convergent and may-divergent: depending on which receiver communicates with the sender, the result is 
the successful process $\nu x,y_1.(\inputxy{x}{y_1}.0\PAR\tStop)$
or the must-divergent process  $\nu x,y_2.(\inputxy{x}{y_2}.\tStop)$.
The $\CH$-process $\tau_0(P)$ reduces after several steps to the process
{\normalfont
$$\begin{array}{l}
\multicolumn{1}{@{}l@{}}{\nu stop,chanx,chany_1,chany_2,chanz,checkx,x,y_1,y_2,z.(}\\
\MVAR{chanx}{(z,checkx)} \PAR 
\EMPTYMVAR{chany_1} \PAR 
\EMPTYMVAR{chany_2} \PAR 
\EMPTYMVAR{chanz}  
\PAR
\MVAR{checkx}{\texttt{\rm ()}}
\PAR 
\MVAR{\istop}{\texttt{\rm ()}}
\\\PAR
x{=}\tChan\,{chanx} \PAR 
z{=}\tChan\,{chanz} 
\PAR 
y_1{=}\tChan\,{chany_1} \PAR 
y_2{=}\tChan\,{chany_2} 
\PAR 
\MTHREAD{}
                  \tputMVar~\istop~()
\\
\PAR 
\THREAD{}
 \bfdo~\{\begin{array}[t]{@{}l@{}}
         \tputMVar~checkx~();\treturn~\texttt{\rm ()}\}
 \end{array}\\
\PAR 
 \THREAD{}
 \bfdo~\{
 \begin{array}[t]{@{}l@{}}
    (y_1,checkx) \leftarrow \ttakeMVar~chanx; \ttakeMVar~checkx; \treturn~\texttt{\rm ()}\}
 \end{array}
\\
\PAR 
 \THREAD{}
 \bfdo~\{
 \begin{array}[t]{@{}l@{}}
 (y_2,checkx) \leftarrow \ttakeMVar~chanx;
%  \\
 \ttakeMVar~checkx;
%  \\
 \ttakeMVar~\istop\})
 \end{array}
\end{array}
$$}
Now the first thread (which is the translation of sender $\outputxy{x}{z}.0$) is blocked,
since it tries to fill the full MVar $checkx$. The second thread (the encoding of $\inputxy{x}{y_1}.0$) and the third thread (the encoding of $\inputxy{x}{y_2}.\tStop$) race for emptying the MVar $chanx$.
If the second thread wins, then it will fill the MVar $checkx$ which is then emptied by the first thread, and all other threads are blocked forever.
If the third thread wins, then it will fill the MVar $checkx$ which is then emptied by
the first thread, and then the second thread will empty the MVar $\istop$. This allows the main-thread to fill it, resulting in a 
% The result is the 
successful process.
% $$\begin{array}{l}
% \multicolumn{1}{@{}l@{}}{\nu chanx,chany_1,chany_2,chanz,checkx,x,y_1,y_2,z.(}\\
% \MVAR{chanx}{(z,checkx)} \PAR 
% \EMPTYMVAR{chany_1} \PAR 
% \EMPTYMVAR{chany_2} \PAR 
% \EMPTYMVAR{chanz} 
% \\
% \PAR
% \EMPTYMVAR{checkx}
% \PAR 
% \MVAR{\istop}{\texttt{()}}
% \PAR
% x{=}\tChan\,{chanx} \PAR 
% z{=}\tChan\,{chanz} 
% \\
% \PAR 
% y_1{=}\tChan\,{chany_1} \PAR 
% y_2{=}\tChan\,{chany_2} 
% \\
% \PAR 
% \THREAD{}
%  \treturn~\texttt{()}
% \\
% \PAR 
%  \THREAD{}
%  \bfdo~\{
%  \begin{array}[t]{@{}l@{}}
%     (y_2,checkx) \leftarrow \ttakeMVar~chanx;
%  \\
%  \ttakeMVar~checkx;
%  \\
%  \treturn~\texttt{()}\}
%  \end{array}
% \\
% \PAR 
%  \THREAD{}
%  \treturn~\texttt{()}
% \\
% \PAR 
% \MTHREAD{}\treturn~\texttt{()})
% \end{array}
% $$
% \color{black}
\end{example}

% We now define what it means for the translation $\tau$ to be compositional, adequate, or fully abstract.
% We 
For the following definition of $\tau$ being compositional, adequate, or fully abstract,
we  
adopt the view that  $\tau$ is 
a translation from $\pistop$ into the $\CH$-language with a special initial evaluation context  $\Couttau$. 

% \todo{
%  \manfredcomment{
%  Irgendwie komisch: Man sollte $\leq$ unten fuer Prozesse und nicht fuer expressions definieren, oder?   UND 
%  subset:  sollte sich auf gleiche Objekte beziehen; prozesse und ncht gemixed.}
% \david{umgekehrt: alles für Ausdrücke machen, da das Cout-Loch ein expression-loch ist. Hab es mal so gemacht.\\
% subset: das normale $\leq_c,\sim_c$ ist für beides (Prozesse und Ausdruecke) definiert.}
% }´
\begin{definition}
 Let 
%  context $\Couttau$ for translation $\tau_0$
%  is defined as
$
\begin{array}{@{}l@{}l@{}}
\Couttau  := \nu \istop. \UMTHREAD{\bfdo}~\{&\istop \leftarrow \tnewMVar~();
\tforkIO~[\cdot];
\tputMVar~\istop~()\}.
\end{array}
$ 
Variants $\maycontau,\mustcontau$ of 
% $\maycon,\mustcon$
 may- and should-convergence
 of expressions $e$
within the context $\Couttau$
in $\CH$ are defined as
% 
%   \begin{center}
      $e \maycontau~ \text{iff}~  \Couttau[e]  \maycon$~and~$e \mustcontau  ~\text{iff}~  \Couttau[e]  \mustcon$. 
%   \end{center}
%        % }\\
% 
% \manfredcomment{umgestellt usw: weiss nicht wozu das closed gut war\ldots.\\ 
% Ah: das closed ist wichtig im pi-calc, aber in CH sollte man es nicht fordern. Es mus sich danach richten, was man 
% bei convergence equivalence braucht. } 
% We define
The relation $\sim_{c,\tau_0}$ is defined by 
$\sim_{c,\tau_0} := \leq_{c,\tau_0} \cap \geq_{c,\tau_0}$, where
% $\le_{c,\tau_0}$ is defined as
% \begin{center}
% \begin{tabular}{l@{~}l}
\mbox{$e_1 \le_{c,\tau_0} e_2$}  iff
   $\forall C: \text{if } \FV(C[e_1], C[e_2]) \subseteq \{\istop\},  \text {then}$
%    \\
%                   &
%                   $\!\!\quad     
                   $C[e_1]\maycontau \implies C[e_2]\maycontau
                       \mbox{ and }  C[e_1]\mustcontau \implies C[e_2] \mustcontau$.
% \end{tabular}
% \end{center}
% \manfredins{
%   $$\begin{array}{l@{~}c@{~}l}
%       e_1 \le_{c,\tau_0} e_2 & \text{iff} & \forall C: \text{if } \FV(C[e_1], C[e_2]) \subseteq \{stop\},  \text {then } \\
%                       && \hspace*{1cm} C[e_1]\maycontau \implies C[e_2]\maycontau
%                        \mbox{ and }  C[e_1]\mustcontau \implies C[e_2] \mustcontau 
% %                        \\[1.1ex]
% %        e_1 \sim_{c,\tau_0} e_2 &  \text{iff}  &    e_1 \le_{c,\tau_0} e_2  \mbox{ and }  e_1 \le_{c,\tau_0} e_2. 
%   \end{array}
%    $$
%    }
%    We define the relations $\le_{c,\tau_0}$ and $\sim_{c,\tau_0}$ on $\CH$-expressions:  %% $\CHF$-expressions:
% % as follows:
%   $$\begin{array}{l@{~}c@{~}l}
%       e_1 \le_{c,\tau_0} e_2 & \text{iff} & \forall C \mbox{ such that } \Couttau[C[e_1]], \Couttau[C[e_2]] \mbox{ are closed}: \\
%                        & &  \Couttau[C[e_1]]\maycon \implies \Couttau[C[e_2]]\maycon
%                        \mbox{ and }   \Couttau[C[e_1]]\mustcon \implies \Couttau[C[e_2]]\mustcon \\[1.1ex]
%        e_1 \sim_{c,\tau_0} e_2 &  \text{iff}  &    e_1 \le_{c,\tau_0} e_2  \mbox{ and }  e_1 \le_{c,\tau_0} e_2. 
%   \end{array}
%    $$
\end{definition}
 
  Since $\le_{c,\CH}$ is a subset of $\le_{c,\tau_0}$,
%   and also  $\sim_{c,\CH}$ is a subset of $\sim_{c,\tau_0}$,
we often can use the more general relations for reasoning.
% and
% in particular 
% can reuse results from $\CH$.  

% \davidins{
%   and will be investigated in the paper:\\
% Versuch, anders zu formatieren:
%  $\begin{array}{@{}llclr}
%   \forall P,P' \in \pistop^c:& P \sim_c P' & \Leftarrow & \tau_0(P) \sim_c  \tau_0(P')& \mbox{(closed-adequacy of } \tau_0)\\ 
%    \forall P,P' \in \pistop^c:&  P \sim_c P'  &\iff&  \multicolumn{2}{l}{\tau_0(P) \sim_c  \tau_0(P') \hspace*{2mm} \mbox{(closed-full abstraction of } \tau_0)}\\
%   
%   \forall P,P' \in \pistop^c:& \mbox{\hspace*{6mm}} P\maycon &\iff&   \tau_0(P)\maycon &  \\  
%    %  &  \mbox{ \hspace*{1cm} and}& P\mustcon &\iff&   \tau_0(P)\mustcon  & \hspace*{-5mm}\mbox{(convergence-equivalence of }  \tau_0)\\
%       & \mbox{and~}  P\mustcon &\iff&   \multicolumn{2}{l}{\tau_0(P)\mustcon  \hspace*{1cm} \mbox{(convergence-equivalence of }  \tau_0)}\\
%    \multicolumn{3}{l}{\forall P \in \pistop.~  \forall C \in \pistopC.~ \forall \xi \in \{\maycon,\mustcon\}:}\\
%     & \Cout[\tau(C[P])] \xi & \iff & \Cout[\tau(C)[P]] \xi& \mbox{(compositionality of } \tau)\\   
%    \forall P,P' \in \pistop: & P \sim_c P' & \Leftarrow & \tau(P) \sim_{c,\tau_0} \tau(P')& \mbox{(adequacy of } \tau)\\ 
%  \forall P,P' \in \pistop: & P \sim_c P' & \iff & \tau(P) \sim_{c,\tau_0} \tau(P')& \mbox{(full abstraction of } \tau)\\ 

%%\manfredcomment{Hier Reklame und Erklaerung fuer adequate\ldots.. ?}

\begin{definition}\label{def:adequate-and-fullyabstract}
Let $\pistopC$ be the  contexts of $\pistop$.  %% closed
We define the following properties for $\tau_0$ and $\tau$ 
\reportOnly{(see \cite{schmidt-schauss-niehren-schwinghammer-sabel-ifip-tcs:08,schmidtschauss-sabel-niehren-schwing-tcs:15})}\paperOnly{(see \cite{schmidtschauss-sabel-niehren-schwing-tcs:15} for a general framework of properties of translations under observational semantics)}. 
%
% \ignore{
% \noindent For closed processes, we say that translation $\tau_0$ is 
% \begin{description}
%  \item[\emph{closed-adequate}] iff 
%  $\text{for all closed processes }P,P' \in \pistop^c$: $\tau_0(P) \sim_c  \tau_0(P')  \implies  P \sim_c P'$,
%  \item[\emph{closed-fully abstract}] iff 
%  $\text{for all closed }P,P' \in \pistop^c{}$:
%   $P \sim_c P'  \iff  \tau_0(P) \sim_c  \tau_0(P')$, and
%  \item[\emph{convergence-equivalent}] iff for all closed $P \in \pistop^c$: 
%   $P\maycon \iff   \tau_0(P)\maycon$ and ${P\mustcon {\iff}   \tau_0(P)\mustcon}$.
% \end{description}
% }
% For open processes $P,P'$, we say that translation $\tau$ is 
% \begin{description}
%  \item[\emph{compositional upto $\{\maycon,\mustcon\}$}] iff for all
%     $P \in \pistop$, all  $C\in\pistopC$, and all $\xi \in 
%     \{\maycon,\mustcon\}$: $\Couttau[\tau(C[P])] \xi \iff  \Couttau[\tau(C)[\tau(P)]] \xi$,
%   \item[\emph{adequate}] iff for all processes
%   $P,P' \in \pistop$: $\tau(P) \sim_{c,\tau_0} \tau(P') \implies  P \sim_c P'$, and
%   \item[\emph{fully abstract}] iff for all processes
%   $P,P' \in \pistop$: $P \sim_c P' \iff  \tau(P) \sim_{c,\tau_0} \tau(P')$.
% \end{description}
For open processes $P,P'$, we say that translation $\tau$ is 
\begin{description}
\item[{convergence-equivalent}] iff for  all   %\davidins{
    $P \in \pistop$:
    $P\maycon \iff \tau(P)\maycontau$ and 
    $P\mustcon \iff \tau(P)\mustcontau$,
 \item[{compositional upto $\{\maycontau,\mustcontau\}$}] iff for all
    $P \in \pistop$, all  $C\in\pistopC$, and all $\xi \in 
    \{\maycontau,\mustcontau\}:$\\
    $\text{if } \FV(C[P]) \subseteq \{stop\}, \text{ then } \tau(C[P]) \xi \iff \tau(C)[\tau(P)] \xi$,
  \item[{adequate}] iff for all processes
  $P,P' \in \pistop$: 
  $\tau(P) \le_{c,\tau_0} \tau(P') \implies  P \le_c P'$, and
  \item[{fully abstract}] iff for all processes
  $P,P' \in \pistop$: 
%   \\
  $P \le_c P' \iff  \tau(P) \le_{c,\tau_0} \tau(P')$.
\end{description}
 \end{definition}
%%% \marginpar{ \manfredcomment{Achtung $\le_c$  statt $\sim_c$! gegenuber alter Version} 
%%   }
%  \subsection{Structure of Proofs and Results}
 
%  The plan for the rest of the paper is as follows:
 %%\davidins{..}
 Convergence-equivalence of  translation $\tau_0$ for may- and should-convergence 
 holds. For readability the proof is 
 omitted, but given in the 
 technical report \cite{schmidt-schauss-sabel:frank-60:19}, where we show:
%  \cref{sec:translation-convergencies}
 
% By \cref{prop:may-convergence-equivalence,prop:must-convergence-equivalence} we have: 
\begin{theorem}\label{thm:may-must-equivalence}
  Let  $P \in \Pi_{\tStop}$ be closed.  Then $\tau_0$ is convergence-equivalent for $\maycon$ and $\mustcon$, \ie~$P \maycon$ is equivalent to $\tau_0(P)\maycon$.
  and $P \mustcon$ is equivalent to $\tau_0(P)\mustcon$.
  This also shows convergence-equivalence of $\tau$ \wrt~$\maycontau,\mustcontau$, 
  \ie~$P\maycon \iff \tau(P)\maycontau$ and $P\mustcon \iff \tau(P)\mustcontau$.
\end{theorem}
% 
%  the following two propositions:
%  \begin{proposition}\label{prop:may-convergence-equivalence}\label{prop:mustdiv-equivalence}
%   Let  $P \in \Pi_{\tStop}$ be closed.  Then $\tau_0$ is convergence-equivalent for $\maycon$, \ie~$P \maycon$ is equivalent to $\tau_0(P)\maycon$. 
%  This also implies that $P \mustdiv$ is equivalent to $\tau(P)\mustdiv$.
% \end{proposition}
% 
% \begin{proposition}\label{prop:must-convergence-equivalence}
%   Let  $P \in \Pi_{\tStop}$ be closed.  Then $\tau_0$ is convergence-equivalent for $\mustcon$, \ie~$P \mustcon$ is equivalent to $\tau_0(P)\mustcon$.
% \end{proposition}

\ignore{
%\davidins{\ldots}
We informally describe the main steps of the proof.
Both propositions have an ``if'' and an ``only-if'' part, i.e.~to prove them, one has to show four claims:
\begin{enumerate}
 \item[1.] $\tau_0$ preserves may-convergence, i.e. for closed $P\in\Pi_{\tStop}$: $P\maycon \implies \tau_0(P)\maycon$.
 \item[2.] $\tau_0$ reflects may-convergence, i.e. for closed $P\in\Pi_{\tStop}$: $\tau_0(P)\maycon\implies P\maycon$.
 \item[3.] $\tau_0$ preserves must-convergence, i.e. for closed $P\in\Pi_{\tStop}$: $P\mustcon \implies \tau_0(P)\mustcon$.
 \item[4.] $\tau_0$ reflects must-convergence, i.e. for closed $P\in\Pi_{\tStop}$: $\tau_0(P)\mustcon\implies P\mustcon$.
\end{enumerate}

Actually, for the latter two parts, we show reflection and preservation of may-divergence:
% i.e.:
\begin{enumerate}
 \item[3'.]  $\tau_0$ reflects may-divergence, i.e. for closed $P\in\Pi_{\tStop}$: $\tau_0(P)\maydiv\implies P\maydiv$.
 \item[4'.] $\tau_0$ preserves may-divergence, i.e. for closed $P\in\Pi_{\tStop}$: $P\maydiv \implies \tau_0(P)\maydiv$.
\end{enumerate}
Note that claim 3 is equivalent to 3' and claim 4 is equivalent to claim 4'.

 The proof technique for showing parts 1,2,3',4'
 is to investigate properties of reduction sequences (those that end in a successful process and those that end in a must-divergent process)  
 and to (inductively) apply appropriate (and permitted) rearrangements of reduction sequences and correct program transformations.
%  , i.e. those listed in \cref{prop:det-sr-correct}.
 
 For the given reduction sequences in the $\pistop$-calculus (giving evidence that $P\maycon$, or $P\maydiv$ resp. holds) for proving part 1 and 4', we assume that they are given as steps of $dsr$-reductions (\cref{def:dsr}). This simplifies the proof and is correct due to \cref{theo:pi-dsr}.
  
 A further technical detail is that we use another translation $\sigma$ (and $\sigma_0$) which translates several $\pi$-process components directly into $\CH$-process components (for instance, parallel composition 
 is directly translated into parallel composition), instead
 of translating them into code which generates the components (as $\tau$ does). An intermediate step in the proof is to show equivalence of $\sigma$ and $\tau$  w.r.t. contextual equivalence. 
 
For proving the reflection parts, i.e. part 2 and part 3, the following diagrams sketch the overall idea of the induction proof,
which mainly is on the number $n$ of given reductions.

For part 2 the diagram is

$$ \xymatrix@C=20mm@R=5mm{
P \ar@{-->}[d]_{sr,*}   \ar@{->}[rr]^{\sigma_0} && \sigma_0(P)\ar@{-->}[dl]^{(sr \cup \sim_c),*}\ar@{->}[d]^{sr,n}\\
P_1 \ar@{-->}[r]^{\sigma_0}&  \sigma_0(P_1) (succ.)  & Q_1 (succ.)
}
$$

For part 3, the diagram is

$$\xymatrix@C=25mm@R=5mm{
P \ar@{-->}[d]_{sr,*}   \ar@{->}[rr]^{\sigma_0} && \sigma_0(P)\ar@{-->}[dl]^{(sr \cup \sim_c),*}\ar@{->}[d]^{sr,n}\\
P_1\mustdiv \ar@{-->}[r]^{\sigma_0}&  \sigma_0(P_1) \mustdiv  & Q_1 \mustdiv
}
$$
 
The difference of both parts is the induction base: For part 2, the processes are successful, while for part 3, the processes have to be must-divergent. Both diagrams also 
illustrate the use of correct program transformations to obtain a ``reordered'' sequence of reductions and transformation for $\sigma_0(P)$ such that it can be back-translated 
into a sequence of $sr$-reductions in the $\pistop$-calculus.

%  The proof is divided into 4 parts: preservation and reflection of may-convergence, 
%  and preserveration and reflection of should-convergence.
 %\enddavidins }

 }  %% END IGNORE

%%%%
%%
%%%%

%%}
% \subsection{Main Results for the Translation $\tau_0$}\label{sec:main-results-of-translation}   
We show that the translation is adequate (see \cref{thm:adequate} below).
 The interpretation of this result is that the $\pi$-calculus with the concurrent semantics is semantically represented within $\CH$.
 %
% It does not tell us that the operational behavior is represented, like in   bisimulation results. 
% \marginpar{\davidins{David: Convergence equivalence ja schon, nur nicht als 1 zu 1 correspondence}}. 
This result is on a more abstract level, since it is 
based on the property whether the programs (in all contexts) produce values or may run into failure, or get stuck; or not. Since the $\pi$-calculus does not have a notion of values, 
% like numbers or lists, 
also the translated processes cannot be compared w.r.t. values other than a single form of value.

 The translation $\tau_0$ is not fully abstract (see \cref{thm:non-fully-abstract} below), which is rather natural, since it only means that it is mapped 
 into a subset of the $\CH$-expressions
 and that this is a proper subset w.r.t. the semantics.
%  
% \manfredcomment{  rueckwaerts die closed C[P]-Definition fuer die Relationen nehmen!! und zeigen, dass die aequivalent sind zu den offenen  }
% 
% 
% 
For proving both theorems, we first use  a simple form of a context lemma: 
 \begin{lemma}\label{lemma:D-context-lemma-closed}
  Let $e,e'$ be $\CH$-expressions,    %$\CHF$-expressions, 
  where the only free variable in $e,e'$ is $\istop$. 
  
  Then $\Couttau[e] \leq_c \Couttau[e']$ holds, if and only
  if  $\Couttau[e]\maycon \implies  \Couttau[e']\maycon$  and  $\Couttau[e]\mustcon \implies  \Couttau[e']\mustcon$.
 \end{lemma}

 \begin{proposition}\label{prop:tau-is-compositional}
 The  translation $\tau$  is 
compositional upto $\{\maycontau,\mustcontau\}$.
 \end{proposition} 
\reportOnly{ \begin{proof} 
This follows by checking whether the single cases of the translation $\tau$ are independent of the surrounding context, 
and translate every level independently.
 \end{proof}
 }

%  \begin{theorem}\label{thm:adequacy-fully-abstract}
%   The translation $\tau_0$ is an embedding, \ie~it is closed-adequate and closed-fully abstract: for closed  $P_1, P_2 \in \Pi_{\tStop}$, 
%   the relation  $\tau_0(P_1) \sim_c \tau_0(P_2)$  holds iff  $P_1 \sim_c P_2$. 
% \end{theorem}
% \begin{proof}
% 
%     The implication $P \sim_c P' \implies \tau_0(P) \sim_c \tau_0(P')$ follows
%     from \cref{lemma:D-context-lemma-closed}, since $\tau_0$ produces closed processes that are in context $\Couttau$.
%     For the other direction (closed-adequacy of $\tau_0$), we additionally require that for all closed $\pistop$-processes 
%     $P,P'$, contexts $C$, and $\xi\in\{\maycon,\mustcon\}$ the implication $(P\xi \iff P'\xi) \implies (C[P]\xi \iff C[P']\xi)$ holds.
%     This can be proved by standard methods and the fact, that a closed $\pistop$-process cannot communicate with other processes, 
%     even if it is replicated.
%    \end{proof}
%  

We show that the translation $\tau$ 
% that is $\tau_0$ without an initialization, 
transports $\pistop$-processes into $\CH$,   %$\CHF$, 
   such that adequacy holds. Thus 
%    is a strong  statement,
%    since it 
%    shows that 
   the translated processes also correctly mimic the behavior of the original $\pistop$-processes when plugged into contexts.
   If the translated open processes cannot be distinguished by $\le_{c,\tau_0}$, \ie~there is no test that detects a difference \wrt may- and should-convergence, then the original processes are equivalent in the $\pi$-calculus.
%Also, the image of $\pistop$ within %%$\CHF$ 
%     $\CH$ behaves much the same  as the original $\pistop$-processes.
% 
However, this open translation is not fully abstract, which means that there are    %$\CHF$-contexts 
  $\CH$-contexts (not in the image of the translation) 
  that can see and exploit too much of the 
details of the translation.
% \manfredcomment{das ist eine Vermutung, dass es falsch ist. Gegenbeispiel suchen\ldots}

% We state and prove the main result of our paper.

\begin{theorem}\label{thm:adequate}
The translation $\tau$ is adequate. 
\end{theorem}
\begin{proof}
We prove the adequacy for the preorder $\le_{c,\tau_0}$, %%%   and $\le_c$.
for  
% the contextual equivalence 
$ \sim_{c,\tau_0}$ and $\sim_c$ the claim follows by symmetry.
Let $P,P'$ be $\pistop$-processes, such that $\tau(P) \le_{c,\tau_0} \tau(P')$.  
We show that  $P \le_c P'$.
We use \cref{lem:conv-open-equiv-closed} to restrict considerations to closed $C[P], C[P']$ below.
Let $C$ be a context in $\pistop$, 
such that $C[P], C[P']$  are closed and  $C[P]\maycon$. 
Then $\tau_0(C[P]) = \Couttau[\tau(C[P])]$.   
Closed convergence-equivalence implies $\Couttau[\tau(C[P])]\maycon$.
By \cref{prop:tau-is-compositional}.
we have $\Couttau[\tau(C)[\tau(P)]]\maycon$.
Now $\tau(P) \le_{c,\tau_0} \tau(P')$ implies
   $\Couttau[\tau(C)[\tau(P')]]\maycon$, which is the same as  $\Couttau[\tau(C[P'])]\maycon$ using  \cref{prop:tau-is-compositional}.
   Closed convergence-equivalence implies $C[P']\maycon$. 
   The same arguments hold for $\mustcon$ instead of $\maycon$. 
%   
%    The computation can be done for every context $C$ that satisfies the conditions. 
   In summary, we obtain  $P \le_c P'$.
\end{proof}

%Hab Deine Anmerkung mal umesetzt:
\begin{theorem}\label{thm:non-fully-abstract}
The translation $\tau$ is not fully abstract,
% \davidins{
but it is fully abstract on closed processes, i.e.~for closed processes $P_1, P_2 \in \Pi_{\tStop}$, 
   we have 
   $P_1 \leq_c P_2  \iff  \tau(P_1) \leq_{c,\tau_0} \tau(P_2)$.
%    }
\end{theorem}
\begin{proof}
The first part holds, since an open translation can be closed  by a context without initializing the $\nu$-bound MVars. 
For $P = \bar{x}(y).\tStop \PAR x(z).\tStop$, we have $P \sim_c \tStop$
% in the $\pi$-calculus,
% (see  \cref{ex:xyStopPARbarxzStopIsSTOP})
 but
$\tau(P) \not\sim_{c,0} \tau(\tStop)$:
% in  %   in $\CHF$: 
%  $\CH$:  
let $\PCtxt$ be a context
that does not initialize the MVars for $x$ (as the translation does). Then $\PCtxt[\tau(P)]\mustdivtau$, but $\PCtxt[\tau(\tStop)]\mustcontau$.
Restricted to closed processes,  full abstraction holds:      $P_1 \leq_c P_2 {\implies} \tau(P_1) \leq_{c,\tau_0} \tau(P_2)$ follows
from \cref{lemma:D-context-lemma-closed},
    since $\tau_0$ produces closed processes in context $\Couttau$.
    \Cref{thm:adequate} implies the other direction.  
    \end{proof}
%of the following theorem 

% \marginpar{\davidins{Referee: Why not replacing Thm 4.9 with a full abstraction theorem for closed processes?}}

% \ignore{ manfred: einfach mal loeschen\ldots
%  \begin{theorem}\label{thm:fully-abstract-for-closed}
%   For closed processes $P_1, P_2 \in \Pi_{\tStop}$, 
%    we have 
%    $P_1 \leq_c P_2  \iff  \tau(P_1) \leq_{c,\tau_0} \tau(P_2)$.
% %    holds. 
% \end{theorem}
% }
% \begin{proof}
% %     For closed $P_1,P_2$, 
%     \Cref{lemma:D-context-lemma-closed}  implies
% %     The implication 
%      $P_1 \leq_c P_2 {\implies} \tau(P_1) \leq_{c,\tau_0} \tau(P_2)$,
% %      follows
% %     from \cref{lemma:D-context-lemma-closed}, 
%     since $\tau_0$ produces closed processes that are in context $\Couttau$.
%     The other direction follows from Main Theorem~\ref{thm:adequate}.
%    \end{proof}
 
%%% TEXEXPAND: END FILE ./translation-using-CH.tex
%%% TEXEXPAND: INCLUDED FILE MARKER ./cases-translations.tex
\section{Translations with Global MVars}\label{sect:global-translations}
%%\davidins{}
In this section we investigate translations that do not use private MVars, but use a fixed number of global
MVars. We first motivate this investigation. The translation $\tau$ is quite complex and thus we want to figure out whether there are simpler translations. 
A further reason is that 
$\tau$ is not optimal, since  it generates one MVar per communication which can be garbage-collected after the communication, however, generation and garbage collection 
require resources and thus the translation $\tau$ may be inefficient in practice. 
% \marginpar{\davidins{\bfseries mehr Motivation fuer globale MVars}}

To systematically search for small global translations we implemented an automated tool.
It searches for translations with global MVars
(abstracting from a lot of other aspects of the translation) and tries to refute the 
correctness. As we show, most of the small translations are 
shown as incorrect by our tool. Analyzing correctness of the remaining translations can then be done by hand.

% In this section we also report on an automated tool 
% that we implemented and 
% that searches for those translations and tries to refute their correctness. ´
We only consider the aspect of how to encode the synchronous message passing of the $\pi$-calculus, the other aspects (encoding parallel composition, replication and the $\tStop$-constant) are not discussed and we assume that they are encoded as before (as the translation $\tau$ did). 
We also keep the main idea to translate a channel of the $\pi$-calculus into $\CH$ by representing it 
as an object of a user-defined data type \texttt{Channel} that consists of
% % \begin{enumerate}
%  \item 
 an MVar for transferring the message 
 (which again is a \texttt{Channel}), and
%  \item 
 additional MVars for implementing a correct synchronization mechanism.
% \e¸nd{enumerate}
For the translation $\tau$,
%in the previous section, 
we used a private MVar 
(created by the sender, 
% for every communication action 
 and transferred together with the message).
% In this section 
Now 
we investigate translations where 
this mechanism is replaced by one or several public MVars, which are created once together with the channel object.
To restrict the search space for translations, only the synchronization mechanism of MVars (by emptying and filling them) is used, but we forbid to transfer specific data (like numbers etc.).
Hence, we restrict these MVars (which we call \emph{check-MVars}) to be of type 
\texttt{MVar~()}.
Such MVars are comparable to binary semaphores, where filling and emptying correspond to operations signal and wait.
In summary, we analyze translations of $\pi$-calculus channels into a $\CH$-data type $\texttt{Channel}$ defined in Haskell-syntax
as
% with a data type definition
% in Haskell of the following form:
$$\texttt{data Channel\,=\,Chan\,(MVar\,Channel)}\,(\texttt{MVar\,()})\ldots(\texttt{MVar\,()})
 $$
A $\pi$-calculus channel $x$  is 
represented as a $\CH$-binding
$x = \texttt{Chan}~content~check_1~\ldots~check_n$
where $content$, $check_1,\ldots,check_n$ are appropriately 
initialized ({\ie}~empty) MVars. 
The MVars are public (or global), since all processes which know  $x$ have access to the components of the channel. 
After fixing this representation of a $\pi$-channel in $\CH$, 
the task is to translate the 
 input- and output-actions $\inputxy{x}{y}$ and  $\outputxy{x}{z}$ into $\CH$-programs such that the interaction reduction is performed correctly and synchronously.
%  \footnote{We omit 
%       the translation of the $\nu$-operator at this moment, but  clearly, it  has to construct and initialize the appropriate MVars. 
%       In \cref{def:induced-translation} the full translation will be given.}.
% 
We call the translation of $\inputxy{x}{y}$, 
the \emph{receiver} (program)
% (the receiver, for short), 
and the translation of $\outputxy{x}{z}$ the \emph{sender} (program).
% (the sender).
As a simplification, we restrict the 
allowed operations
of the sender and receiver allowing only the operations:
\begin{description}
 \item[$\putS$:] The sender puts its message into the $contents$-MVar of the channel.  
  It represents the expression 
  $\tcase_\texttt{Channel}~x~\tof~(\tChan~c~a_1~\ldots~a_n \casepf~\tputMVar~c~z~\tbindthen~e)$
  in $\CH$ where $e$ is the remaining program of the sender. The operation occurs exactly once in the sender program.  We write it as $\putS_x~z$, or as $\putS$, if $x$ and $z$ are clear.
%  from the context.
 \item[$\takeS$:] The receiver takes the message from the $contents$-MVar of channel $x$ and 
 replaces name $y$ by the received name in the subsequent program.
 The operation occurs exactly once in the receiver program. 
 We write it as $\takeS_x~y$, or as $\takeS$,
  if $x$ and $y$ are clear. It represents the $\CH$-expression
 $\tcase_\texttt{Channel}~x~\tof~(\tChan~c~a_1~\ldots~a_n \casepf~\ttakeMVar ~c~\tbind~\lambda y.e)$
 where $e$ is the remaining program of the receiver.
  %% \davidins{}
 In $\bfdo$-notation, we write $\bfdo~\{y \leftarrow \takeS_x;e\}$ to abbreviate the above $\CH$-expression.
 
 \item[$\putCh$ and $\takeCh$:] The sender and the receiver may synchronize on a check-MVar $check_i$ by putting $()$ into it or by emptying the MVar. 
 These operations are written as $\putCh_x^i$ and $\takeCh_x^i$, or also
%  even more abstractly as
 as $\putCh^i, \takeCh^i$ if the name $x$ is clear.
%  from the context.
 We write $\putCh$ and $\takeCh$ if there is only one check-MVar.
  Let $e$ be the remaining program of the sender or receiver. Then $\putCh_x^i$ represents
the $\CH$-expression
  $\tcase_\texttt{Channel}~x~\tof~(\tChan~c~a_1~\ldots~a_n \casepf~\tputMVar~a_i~()~\tbindthen~e)$
%   is represented by $\putCh_x^i$
  and  $\takeCh_x^i$ represents 
  the expression
%%  $$\tcase~x~\tof~(\tChan~c~a_1~\ldots~a_n \casepf~\tputMVar~a_i~()~\tbindthen~e)$$
%  \manfredins{
$\tcase_\texttt{Channel}~x~\tof~(\tChan~c~a_1~\ldots~a_n \casepf~\ttakeMVar~a_i~\tbindthen~e).$
%  
%  is represented by
%  represents the $\CH$-program

%  }
\end{description}
  
We restrict our search for translations to the case that the sender and the receiver programs are sequences of the above operations,
assuming that they are independent of the channel name $x$. With this 
restriction, we can abstractly write the translation of the sender and the receiver as a pair of sequences, where only 
$\putS,\takeS,\putCh^i$ and $\takeCh^i$ operations are used. We make some more  restrictions:

\begin{definition}\label{def:translation}
  Let $n > 0$ be a number of check-MVars.
  A {\em standard global synchronized-to-buffer translation} 
  (or {\em gstb-translation}) 
  is 
%   represented as 
  a pair $(T_\mathit{send},T_\mathit{receive})$ of a send-sequence $T_\mathit{send}$ and a receive-sequence $T_\mathit{receive}$ consisting 
  of $\putS$, $\takeS$, $\putCh^i$ and $\takeCh^i$ operations, where 
%   \begin{itemize}
%     \item 
    the send-sequence contains $\putS$ once,  the receive-sequence contains $\takeS$ once, and
%      \item 
     for every $\putCh^i$-action in $(T_\mathit{send},T_\mathit{receive})$, there is also a $\takeCh^i$-action in $(T_\mathit{send},T_\mathit{receive})$.
%    \end{itemize}
   W.l.o.g., we  assume that in the send-sequence the indices $i$ are ascending. I.e. if  
     $\putCh^i$ or $\takeCh^i$ is  before $\putCh^j$ or $\takeCh^j$, then $i < j$ holds.
\end{definition}

We often say translation instead of gstb-translation, if this is clear from the context.

\begin{definition}\label{def:induced-translation}
Let $T=(T_\mathit{send},T_\mathit{receive})$ be a  gstb-translation.
We write $T_\mathit{send}^{x,y}$ for the program $T_\mathit{send}$ 
instantiated for   $\outputxy{x}{y}$,
\ie~$\putS$ is $\putS_x~y$, and all other operations are
indexed with  $x$.
% , i.e $\takeCh^i$ becomes $\takeCh_x^i$ and
% $\putCh^i$ becomes $\putCh_x^i$.
We write $T_\mathit{receive}^{x,y}$ for the program $T_\mathit{receive}$ 
instantiated for $\inputxy{x}{y}$,
\ie~$\takeS$ is $\takeS_x~y$, and all other operations are indexed with $x$.
% ($\takeCh^i$ becomes $\takeCh_x^i$, $\putCh^i$ becomes $\putCh_x^i$).
% 
%%% TEXEXPAND: INCLUDED FILE MARKER ./figure-induced-trans.tex
\begin{figure*}[tp]
$\begin{array}{@{}l@{\,}l@{\,}l@{}}
\tauNullInduced(P) &= &  \UMTHREAD{\bfdo~\{}
                 \begin{array}[t]{@{}l@{}}
                  \istop \leftarrow \tnewMVar~();  
%                   \\
                  \tforkIO~\tauInduced(P);  
%                   \\
                  \tputMVar~\istop~()\}
                 \end{array}
            \\
    \tauInduced(\outputxy{x}{y}.P)  & =& \bfdo~\{T_\mathit{send}^{x,y};\tauInduced(P)\}
            \\[1ex]
    \tauInduced(\inputxy{x}{y}.P) & =& \bfdo~\{T_\mathit{receive}^{x,y};\tauInduced(P)\}
            \\
% % % %  
    \tauInduced(P \PAR Q) &  =& \bfdo~\{\tforkIO~\tauInduced(Q);\tauInduced(P)\}
            \\\\[-2.8ex]
     \tauInduced(\nu x.P) &  =&  \bfdo\,\{
        \begin{array}[t]{@{}l@{}}
        \mathit{contx} \leftarrow \tnewEmptyMVar;
%         \\
        \mathit{checkx}_{1} \leftarrow \tnewEmptyMVar;
%         \\
        \ldots;
%         \\
        \mathit{checkx}_{n} \leftarrow \tnewEmptyMVar;        
         \\
        \tletr~ x = \tChan\,\mathit{contx}~\mathit{checkx}_{1}\ldots\mathit{checkx}_{n}\tin~\tauInduced(P)\}\\
     \end{array}
\\
\tauInduced(0) &  = & {\treturn~()} 
\\
\tauInduced(\tStop) &  = & {\ttakeMVar~\istop}  
\\
\tauInduced(!P) & = & {\tletr~ f = \bfdo~\{\tforkIO~\tauInduced(P); f\}~\tin~ f}  
\end{array}$
%%%%\davidins{}
\caption{Induced translations $\tauInduced$ and $\tauNullInduced$ where  $T=(T_\mathit{send},T_\mathit{recieve})$ uses $n$  check-MVars \label{fig:def-tauinduced}}
\end{figure*}
%%% TEXEXPAND: END FILE ./figure-induced-trans.tex
The \emph{induced translations} $\tauNullInduced$ and $\tauInduced$ of $(T_\mathit{send},T_\mathit{receive})$ are defined in \cref{fig:def-tauinduced}.
% , where $n$ is the number of check-MVars used by the translation.
% 
 \end{definition} 
 
% \davidins{..}
The induced translations are defined similar to the translations $\tau_0$ and $\tau$, where the differences are the representations of the channel. The translation of $\nu x$, $\inputxy{x}{y}$, and $\outputxy{x}{y}$ is different, but
the other cases remain the same.
Since $\tauNullInduced(P) = \Couttau[\tauInduced(P)]$
and by the same arguments as in  \cref{thm:adequate}, we have:

\begin{proposition}\label{prop:adequateInduced}
If $\tauInduced$ is closed convergence-equivalent, 
then $\tauInduced$ is adequate.
%%  and on closed processes t is fully-abstract.
\end{proposition}

An {\em execution} of a translation $(T_\mathit{send},T_\mathit{receive})$ for name $x$ is the simulation of the \emph{abstract} program, \ie~a program that starts with empty MVars $x$, $ x_1,\ldots,x_n$, and is an interleaved sequence 
of actions from the send and receive-sequence $T_\mathit{send}$ and $T_\mathit{receive}$, respectively.

% \subsection{Classifying Translations}
To speak about the translations we make further classifications: We say that
 a translation allows {\em multiple uses},
 if a check-MVar is used more than once, \ie~the sender or receiver may contain $\takeCh^i$ or $\putCh^i$ more than once for the same $i$.
 A translation has the {\em interprocess check restriction}, if for every $i$: 
 $\takeCh^i$ and $\putCh^i$ do not occur both in $T_\mathit{send}$, and also not both in $T_\mathit{receive}$.
% \end{description}
% 
% 
% We define some properties of gstb-translations:
\begin{definition}
 A translation $T=(T_\mathit{send},T_\mathit{receive})$ according to \cref{def:translation} is 
\begin{itemize} 
\item \emph{executable}
   if there is a deadlock free 
%    interleaving 
execution of $T$; 
% (T_\mathit{send},T_\mathit{receive})$.   %%        of all actions
\item  \emph{communicating}, if  $T_\mathit{send}$ contains at least one $\takeCh^i$-action; 
\item \emph{overlap-free} if for a fixed name $x$, starting with empty MVars, every 
     interleaved (concurrent) 
      execution 
      of $(T_\mathit{send},T_\mathit{receive})$ 
      cannot be disturbed by starting another execution of $T_{\mathit{send}}$ and/or $T_{\mathit{receive}}$. 
   %   \davidins{..} 
      More formally, let $((s_1;...;s_i); (r_1;....r_j))$
and $((s'_1;...;s'_i); (r'_1;....r'_j))$
be two copies of $(T_{send},T_{receive})$  
for a fixed name $x$.
We call a command $a_k$ from one of the four sequences, 
an $a$-action for $a \in \{s,s',r,r'\}$.
The translation $T$
is overlap-free if every execution of
the four sequences has the property 
that it can be split into a prefix and a suffix
(called parts in the following)
such that one of the following properties holds
\begin{enumerate}
\item  One part contains only $s$- and $r$-actions, 
  and the other part contains only $s'$- and $r'$-actions.
\item  One part contains only $s$- and $r'$-actions
  and the other part contains only $s'$- and $r$-actions. 
\end{enumerate}
% , 
%       after $T_{send}$ and $T_{receive}$ both performed (at least) one action.   %\davidins{}
%  \end{itemize}
 
%  \manfred{reicht das? zB  tau1 = (PTPT\ldots,..) tau2 = (PTPT\ldots,..) kann gestoert werden ohne dass receive anlaeuft. } 
%  \david{
%  Hab den Halbsatz
% ``after $T_{send}$ and $T_{receive}$ both performed (at least) one action.''
% gestrichen.
% %\davidins{}
%  }

%% Here we assume that at the start, all MVars are empty.
\end{itemize}
\end{definition}
   
% \subsection{Simulating Translations to Refute their Correctness}
We implemented a tool to enumerate translations and to test whether each translation preserves and reflects may- and should-convergence
% \manfred{ etwas genauer\ldots und evtl.  auch sagen dass es auch positive Vermutungen macht}
% \manfredins{ 
for a (given) finite set of processes. 
%\davidins{}
Hence, our tool can refute the correctness of translations, but it can also output (usually small) sets of translations which are not refuted and which are promising candidates for correct translations.
The above mentioned parallel execution of $T_{send}$ and $T_{receive}$  is not sufficient to refute 
most of the translations, since it corresponds to the evaluation of 
the $\pi$-program
$\NEW x.(\inputxy{x}{y} \PAR \outputxy{x}{z})$ (which is must-divergent).
% , since no $\tStop$ occurs).
Thus, we apply the translation to a subset of $\pi$-processes, which we view as critical 
and for which we can automatically decide may- and should-convergence (before and after the translation).
We consider only $\pi$-programs of the form 
$(\nu x_1,\ldots,x_n.P)$ where $P$ contains only $0$, $\tStop$, parallel composition, and
input- and output-prefixes.
These programs are replication free and the $\nu$-binders are on the top, and hence terminate. In the following we omit the $\nu$-binders, but always mean them to be present. We also implemented techniques to generates all such programs until a bound on the size of the program is reached.

Our simulation tool\footnote{Available via \href{https://gitlab.com/davidsabel/refute-pi}{https://gitlab.com/davidsabel/refute-pi}.} can execute all possible evaluations of those $\pi$-processes and 
-- since all evaluation paths are finite -- the tool can check for may- and should-convergence of the $\pi$-program.
For the translated program, we do not generate a full $\CH$-program, but generate a sequence of sequences of $\takeS_x,\putS_x,\takeCh_x^i,\putCh_x^i~z$
and $\tStop$-operations by applying the translation to all action prefixes in the $\pi$-program and by encoding $\tStop$ as $\tStop$,
0 into an empty sequence. We get a sequence of sequences, since we have several threads and each thread is represented by one sequence.
For executing the translated program, we simulate the global store (of MVars) and execute all possible interleavings where we check for may- and should-convergence by looking 
whether the $\tStop$ eventually occurs at the beginning of the sequence. 
This simulates the behavior of the real $\CH$-program in a controllable manner.

With the encoding of the sender- and receiver program and a $\pi$-calculus process $P$ we
\begin{enumerate}
 \item translate $P$ with the encodings in the sequence of sequences consisting of $\takeS_x$, 
 $\putS_x$, $\takeCh_x^i$, $\putCh_x^i~z$
and $\tStop$-operations;
 \item simulate the execution on all interleavings; 
 \item test may- and should convergence of the original $\pi$-program $P$ as well as the encoded program (\wrt~the simulation);
 \item compare the convergence before and after the translation.
If there is a difference in the convergence behavior, then $P$ is a counter-example for the correctness of the encodings.
\end{enumerate}

%%\davidins{
% \manfred{Schau nochmal, der Resultatprozess ar mE nicht ganz richtig, und ob die Uebersetzung stimmt.}
% David: ok gecheckt
\begin{example}
Let us consider the gstb-translation
 $(T_\mathit{send},T_\mathit{receive}) = 
 ([\takeCh,\putS], [\putCh,\takeS])$
and  the $\pi$-process
 $P = \nu x,y,z,w(\outputxy{x}{y}.\inputxy{x}{z}.\tStop \PAR \inputxy{x}{w}.0)$.
Our tool recognizes that $P{\uparrow}$ and $P{\Uparrow}$
holds, since $P$ reduces to the must-divergent process 
% $\nu x,z.(\inputxy{x}{z}.0)$ and there are no other reduction possibilities. 
$\nu x,z.(\inputxy{x}{z}.\tStop)$ and there are no other reduction possibilities. 

Applying $(T_\mathit{send},T_\mathit{receive})$
to $P$ yields
 the
 abstract program 
$$q := [[\takeCh_x,\putS_x~y,\putCh_x,\takeS_x~z,\tStop],[\putCh_x,\takeS_x~w]].$$
For $q$, our tool 
% with input $q$ first
inspects all executions.
% (i.e. interleavings of both processes). 
Among them there is the sequence
$$\putCh_x;\takeCh_x;\putS_x~y;\putCh_x;\takeS_x~z;\tStop$$
which can be executed ending in $\tStop$. Thus $q$ is may-convergent, and thus the process
$P$ is a counter-example that refutes the correctness of the translation.
\end{example}
%% }

% \subsection{Translations without Check-MVars}
The case that there is no check-MVar leads to one possible translation $([\putS],[\takeS])$ which means that
$\outputxy{x}{z}$ is translated into $\putS_x~y$
and $\inputxy{x}{y}$ is translated into $\takeS_x~y$. 
This translation is not correct, since for instance the $\pi$-process $\outputxy{x}{z}.\inputxy{x}{y}.\tStop$ is neither may- nor should-convergent, but the translation
(written as an abstract program)
is  $[[\putS_x~z,\takeS_x~y,\tStop]]$.
I.e., it consists of one process which is may- and should-convergent (since $\putS_x~z;\takeS_x~y;\tStop$ is the only evaluation sequence and its execution ends in $\tStop$).
Note that the translation into $\CH$ will generate two threads: the main threads that will wait until the MVar $\istop$ is filled, and a concurrent thread that will do the above operations.

% , since the put- and the take-operation can be done sequentially.

\subsection{Translations with Interprocess Check Restriction }
We  consider translations with the interprocess check restriction 
(each $\takeCh^i$ and $\putCh^i$ must be distributed between the sender and the receiver).
There are $n! \cdot 2^n \cdot (n+1)^2$ different translations (for $n$ check-MVars).
%, which for $n = 1$ results in $8$, for $n = 2$ in 72,  
%and for  $n = 3$ in 768.
% 
% 
\begin{table*}[t]
\centering\begin{tabular}{|l|l|@{~}c@{~}|@{~}c@{~}|@{~}c@{~}|@{~}c@{~}|}
Translation (sender,receiver) & Counter-example  ($\pi$-process) 
& \rotatebox{0}{\begin{tabular}{@{}c@{}}\footnotesize $\maycon$ before\end{tabular}}
            &  \rotatebox{0}{\begin{tabular}{@{}c@{}}\footnotesize $\mustcon$ before\end{tabular}}
            &  \rotatebox{0}{\begin{tabular}{@{}c@{}}\footnotesize $\maycon$ after\end{tabular}}
            &  \rotatebox{0}{\begin{tabular}{@{}c@{}}\footnotesize $\mustcon$ after\end{tabular}}
\\
\hline
$([\putCh,\putS],[\takeCh,\takeS])$
 &  $\outputxy{x}{y}.\inputxy{x}{y}.\tStop$
 &N 
 & N
%    & $[[\putCh_x~1,\putS_x~y,\takeCh_x~1,\takeS_x~y,\tStop]]$
 &Y 
 &Y 
\\
   $([\putCh,\putS],[\takeS,\takeCh])$
  & $\outputxy{x}{y}.\inputxy{x}{y}.\tStop$
&N&N
% [[putCh_x 1,putS_x y,takeS_x y,takeCh_x 1,stop]]
&Y&Y
\\
  $([\putS,\putCh],[\takeCh,\takeS])$
 &  $\outputxy{x}{y}.\inputxy{x}{y}.\tStop$
&N&N
%    [[putS_x y,putCh_x 1,takeCh_x 1,takeS_x y,stop]]
&Y&Y
\\
   $([\putS,\putCh],[\takeS,\takeCh])$
&  $\outputxy{x}{y}.\inputxy{x}{y}.\tStop$
&N&N
%    [[putS_x y,putCh_x 1,takeS_x y,takeCh_x 1,stop]]
&Y&Y
\\
 $([\takeCh,\putS],[\putCh,\takeS])$
&   $\outputxy{x}{y}.\inputxy{x}{z}.\tStop~|~\inputxy{x}{w}$
&N&N
%    [[takeCh_x 1,putS_x y,putCh_x 1,takeS_x z,stop],[putCh_x 1,takeS_x w]]
 &Y&N
\\ 
   $([\takeCh,\putS],[\takeS,\putCh])$
 &  $\outputxy{x}{y}.\tStop~|~\inputxy{x}{y}$
 &Y&Y
%    [[takeCh_x 1,putS_x y,stop],[takeS_x y,putCh_x 1]]
 &N&N
\\
  $([\putS,\takeCh],[\putCh,\takeS])$
 &  $\outputxy{x}{y}.\inputxy{x}{z}.\tStop~|~\inputxy{x}{w}$
&N&N
%    [[putS_x y,takeCh_x 1,putCh_x 1,takeS_x z,stop],[putCh_x 1,takeS_x w]]
 &Y&N
\\ 
   $([\putS,\takeCh],[\takeS,\putCh])$
%  is violated by pi-program with number 11:
  & $\outputxy{x}{z}.\outputxy{z}{a}.\tStop~|~\outputxy{x}{w}.\outputxy{w}{a}.\tStop~|~\inputxy{x}{y}.\inputxy{y}{u}$
&Y&Y
%    [[putS_x z,takeCh_x 1,putS_z a,takeCh_z 1,stop],[putS_x w,takeCh_x 1,putS_w a,takeCh_w 1,stop],[takeS_x y,putCh_x 1,takeS_y u,putCh_y 1]]
 &Y&N
\end{tabular}
 \caption{Translations using one check-MVar and with the interprocess check restriction}\label{table-single-check}
% 
% \begin{tabular}{@{}l@{\,}|@{\,}l@{}} 
% \text{Sender} & \text{Receiver}
% \\
% % \hline
% % \multicolumn{2}{l@{}}{\text{Translations without potential overlaps}}
% % \\
% \hline
%  $[\putS,\putCh^1,\takeCh^2,\putCh^3]$ 
% &$[\takeCh^1,\putCh^2,\takeCh^3,\takeS]$
% \\
% $[\takeCh^1,\putS,\takeCh^2,\takeCh^3]$ 
% &$[\putCh^3,\putCh^1,\takeS,\putCh^2]$
% \\
% $[\putCh^1,\putS,\takeCh^2,\putCh^3]$ 
% &$[\takeS,\putCh^2,\takeCh^3,\takeCh^1]$
% \\
% $[\putCh^1,\putCh^2,\takeCh^3,\putS]$  
% &$[\takeCh^2,\putCh^3,\takeS,\takeCh^1]$
% \\
% % \hline
% % \multicolumn{2}{l@{}}{\text{Translations with potential overlaps}}
% % \\
% % \hline
% $[\takeCh^1,\putS,\takeCh^2,\takeCh^3]$ 
% &$[\putCh^1,\putCh^2,\takeS,\putCh^3]$
% \\
% $[\putCh^1,\takeCh^2,\putS,\takeCh^3]$ &
% $[\takeCh^1,\putCh^2,\takeS,\putCh^3]$
% \end{tabular}
% \caption{Non-rejected translations for 3 check-MVars and interprocess check restriction}\label{table-3mvar}
\end{table*}
For a single check-MVar, all 8 translations are rejected by our tool, \cref{table-single-check} shows the translations and the obtained counter examples.
For $2$ check-MVars, our tool refutes all  72 translations.
% Our tool refutes all of them. 
%  Thus there is no correct translation.
Compared to \cref{table-single-check},
there are two further $\pi$-programs used as counterexample.
% They are 
% $\outputxy{x}{w}.\inputxy{x}{y}.\outputxy{x}{y}.\inputxy{y}{q}\PAR\inputxy{x}{z}\PAR\outputxy{x}{z}\PAR\inputxy{w}{p}.\tStop$
% and $\outputxy{x}{y}\inputxy{x}{z}.\outputxy{z}{q}\PAR\inputxy{x}{z}\PAR\outputxy{x}{z}\PAR\inputxy{y}{u}.\tStop$
% 
However, also the programs 
$\outputxy{x}{y}.\tStop \PAR \inputxy{x}{y}$  and   
$\outputxy{x}{y}.\inputxy{x}{z}.\outputxy{z}{q} \PAR \inputxy{x}{z} \PAR$ $\inputxy{x}{z}\PAR \outputxy{x}{z}\PAR \inputxy{y}{u}.\tStop$  suffice to refute all 72 translations.

\begin{theorem}\label{thm:no-gstb-translation-le-three}
There is no valid gstb-translation with the interprocess check restriction
for less than three check-MVars. 
\end{theorem}
% \begin{proof}[Proof (Sketch)] 
% Our simulator refutes all translations for less than three check-MVars. All 72 possibilities of translations using 2 check-MVars can be refuted by
% the programs
% $\outputxy{x}{y}.\tStop \PAR \inputxy{x}{y}$  and   
% %%([[Put "x" "y",Stop],[Take "x" "y" 0]],True,True)
% %%*Main> pruefProgs!!29
% %% ([[Put "x" "y",Take "x" "z" 0,Put "z" "q"],[Take "x" "z" 0],[Take "x" "z" 0],[Put "x" "z"],[Take "y" "u" 0,Stop]],False,False)
% $\outputxy{x}{y}.\inputxy{x}{z}.\outputxy{z}{q} \PAR \inputxy{x}{z} \PAR \inputxy{x}{z}\PAR \outputxy{x}{z}\PAR \inputxy{y}{u}.\tStop$. 
% % suffice to refute all translation variants.
% % The first process is should-convergent 
% % and the latter is neither may- nor should-convergent.
% \end{proof}

 A reason for the failure of translations with less than three check-MVars may be:
%  the following result:
\begin{theorem}\label{thm:no-small-translations}
 There is no executable, communicating, and overlap-free gstb-translation with the interprocess check restriction for $n < 3$. 
\end{theorem}
\begin{proof}
For $n = 1$, we check the translations in \cref{table-single-check}. 
The first four are non-communicating.
For the translation $([\takeCh,\putS],[\takeS,\putCh])$ a deadlock occurs.
% Thus the translation is not executable.
For $([\takeCh,\putS],[\putCh,\takeS])$, after  $\putCh$, $\takeCh$, we can execute $\putCh$ again.
For $([\putS,\takeCh],[\takeS,\putCh])$, after  executing $\putS$, $\takeS$ 
we can execute $\putS$ again.
For $([\putS,\takeCh],[\putCh,\takeS])$, after $\putCh$,$\putS$, $\takeCh$ 
we can execute $\putCh$ again.
For $n = 2$, the simulator finds no executable, communicating, and overlap-free translation:
18 translations are non-communicating, 21  lead to a deadlock, and  33 may lead to an overlap.   
\reportOnly{This check could also be done by hand by scanning all 72 translations.}
\end{proof}
For 3 MVars, our tool rejects 762 out of 768 translations  (using 
the same counter examples as for $2$ check-MVars) 
and the following 6 translations remain:
\begin{center}
\begin{tabular}{@{}l@{\,}c@{\,}l}
$\mathfrak{T}_1$ &$=$& $([\putS,\putCh^1,\takeCh^2,\putCh^3],[\takeCh^1,\putCh^2,\takeCh^3,\takeS])$
\\
$\mathfrak{T}_2 $ &$=$& $ ([\takeCh^1,\putS,\takeCh^2,\takeCh^3],[\putCh^3,\putCh^1,\takeS,\putCh^2])$
\\
$\mathfrak{T}_3 $ &$=$& $ ([\putCh^1,\putS,\takeCh^2,\putCh^3],[\takeS,\putCh^2,\takeCh^3,\takeCh^1])$
\\
$\mathfrak{T}_4 $ &$=$& $ ([\putCh^1,\putCh^2,\takeCh^3,\putS],[\takeCh^2,\putCh^3,\takeS,\takeCh^1])$
\\
$\mathfrak{T}_5 $ &$=$& $ ([\takeCh^1,\putS,\takeCh^2,\takeCh^3], 
[\putCh^1,\putCh^2,\takeS,\putCh^3])$
\\
$\mathfrak{T}_6 $ &$=$& $ ([\putCh^1,\takeCh^2,\putS,\takeCh^3],
[\takeCh^1,\putCh^2,\takeS,\putCh^3])$
\end{tabular}
\end{center}

% shown in \cref{table-3mvar}.

\reportOnly{
\begin{remark}
We exhibit an example, why the translation 
$$(T_\mathit{send},T_\mathit{receive}) = ([\putS,\takeCh],[\takeS,\putCh])$$ is not correct. 
Consider the following $\pistop$-process
$$P = \nu x, z, w, a.(\outputxy{x}{z}.\outputxy{z}{a}.\tStop \PAR \outputxy{x}{w}.\outputxy{w}{a}.\tStop \PAR \inputxy{x}{y}.\inputxy{y}{w}.0)$$
which is should-convergent,  since the two possible reduction sequences starting with $P$ end in successful processes.

Consider the translated $\CH$-program where we use the commands abbreviating the concrete $\CH$-programs,
and where we assume that the program code for creating the MVars and the bindings for the \texttt{Channel}-objects are already
executed:
$$\begin{array}{@{}l@{\,}l@{\,}l@{\,}l@{\,}}
\multicolumn{4}{@{}l}{\nu \ldots.x = \tChan~c_x~chk_x \PAR z = \tChan~c_z~chk_z\PAR w=\tChan~c_w~chk_w \PAR \ldots
}\\
\multicolumn{4}{l}{\PAR \EMPTYMVAR{c_x} \PAR \EMPTYMVAR{chk_x} \PAR \EMPTYMVAR{c_z} \PAR \EMPTYMVAR{chk_z} \PAR \EMPTYMVAR{c_w} \PAR  \EMPTYMVAR{chk_w} \PAR \ldots
% \EMPTYMVAR{c_a} \PAR \EMPTYMVAR{chk_a}
}
\\
 \multicolumn{4}{l}{ \PAR~ \bfdo \{\putS_x~z; \takeCh_x;
   \putS_z~a; \takeCh_z;\ldots \}} \\   %\ttakeMVar ~\istop
  \multicolumn{4}{l}{ \PAR~ \bfdo \{\putS_x~w; \takeCh_x;
   \putS_w~a; \takeCh_w; \ldots\}} \\   %\ttakeMVar ~\istop
   \multicolumn{4}{l}{\PAR~ \bfdo \{y \leftarrow \takeS_x; \putCh_x;
   w\leftarrow\takeS_y; \ldots \}} \\   %\tputMVar ~(\tgetcheck~y); \treturn~()
 \end{array}$$
  This reduces after executing $\putS_x~z$ and $y \leftarrow \takeS_x$
  $$\begin{array}{@{}l@{\,}l@{\,}l@{\,}l@{\,}}
\multicolumn{4}{@{}l}{\nu \ldots.x = \tChan~c_x~chk_x \PAR z = \tChan~c_z~chk_z\PAR w=\tChan~c_w~chk_w \PAR \ldots
}\\
\multicolumn{4}{l}{\PAR \EMPTYMVAR{c_x} \PAR \EMPTYMVAR{chk_x} \PAR \EMPTYMVAR{c_z} \PAR \EMPTYMVAR{chk_z} \PAR \EMPTYMVAR{c_w} \PAR  \EMPTYMVAR{chk_w} \PAR \ldots
% \EMPTYMVAR{c_a} \PAR \EMPTYMVAR{chk_a}
}
\\
 \multicolumn{4}{l}{ \PAR~ \bfdo \{\takeCh_x;
   \putS_z~a; \takeCh_z;\ldots \}} \\   %\ttakeMVar ~\istop
  \multicolumn{4}{l}{ \PAR~ \bfdo \{\putS_x~w; \takeCh_x;
   \putS_w~a; \takeCh_w; \ldots\}} \\   %\ttakeMVar ~\istop
   \multicolumn{4}{l}{\PAR~ \bfdo \{\putCh_x;
   w\leftarrow\takeS_z; \ldots \}} \\   %\tputMVar ~(\tgetcheck~y); \treturn~()
 \end{array}$$   
 The second and third thread make steps: 
   $$\begin{array}{@{}l@{\,}l@{\,}l@{\,}l@{\,}}
\multicolumn{4}{@{}l}{\nu \ldots.x = \tChan~c_x~chk_x \PAR z = \tChan~c_z~chk_z\PAR w=\tChan~c_w~chk_w \PAR \ldots
}\\
\multicolumn{4}{l}{\PAR \MVAR{c_x}{w} \PAR \MVAR{chk_x}{()} \PAR \EMPTYMVAR{c_z} \PAR \EMPTYMVAR{chk_z} \PAR \EMPTYMVAR{c_w} \PAR  \EMPTYMVAR{chk_w} \PAR \ldots
% \EMPTYMVAR{c_a} \PAR \EMPTYMVAR{chk_a}
}
\\
 \multicolumn{4}{l}{ \PAR~ \bfdo \{\takeCh_x;
   \putS_z~a; \takeCh_z;\ldots \}} \\   %\ttakeMVar ~\istop
  \multicolumn{4}{l}{ \PAR~ \bfdo \{\takeCh_x;
   \putS_w~a; \takeCh_w; \ldots\}} \\   %\ttakeMVar ~\istop
   \multicolumn{4}{l}{\PAR~ \bfdo \{
   w\leftarrow\takeS_z; \ldots \}} \\   %\tputMVar ~(\tgetcheck~y); \treturn~()
 \end{array}$$
%  $\begin{array}{@{\qquad}l@{\,}l@{\,}l@{\,}l@{\,}}
% %% \multicolumn{4}{l}{ \nu x,\tsendx,\tcheckx,z,\tsendz,\tcheckz,w,\tsendw,\tcheckw,a,\tsenda,\tchecka.}\\
%  (\MVAR{sendx}{w} &\PAR \MVAR{checkx}{()} &\PAR \EMPTYMVAR{sendz} &\PAR \EMPTYMVAR{checkz} \\
%  ~ \EMPTYMVAR{sendw} &\PAR \EMPTYMVAR{checkw} &\PAR \EMPTYMVAR{senda} &\PAR \EMPTYMVAR{checka} \\
%  \multicolumn{4}{l}{ \PAR~ \bfdo \{\ttakeMVar~ \tcheckx;
%    \tputMVar ~\tsendz~ a; \ttakeMVar ~\tcheckz;\ldots \}} \\   %\ttakeMVar ~\istop
%   \multicolumn{4}{l}{ \PAR~ \bfdo \{\ttakeMVar~ \tcheckx;
%    \tputMVar ~\tsendw~ a; \ttakeMVar ~\tcheckw; \ldots\}} \\   %\ttakeMVar ~\istop
%    \multicolumn{4}{l}{\PAR~ \bfdo \{ 
%    \ttakeMVar ~(\tgetsend~\tsendz); \ldots \}} \\   %\tputMVar ~(\tgetcheck~y); \treturn~()
%  \end{array}$\\
 The second  thread can make an unexpected step: 
   $$\begin{array}{@{}l@{\,}l@{\,}l@{\,}l@{\,}}
\multicolumn{4}{@{}l}{\nu \ldots.x = \tChan~c_x~chk_x \PAR z = \tChan~c_z~chk_z\PAR w=\tChan~c_w~chk_w \PAR \ldots
}\\
\multicolumn{4}{l}{\PAR \MVAR{c_x}{w} \PAR \EMPTYMVAR{chk_x} \PAR \EMPTYMVAR{c_z} \PAR \EMPTYMVAR{chk_z} \PAR \EMPTYMVAR{c_w} \PAR  \EMPTYMVAR{chk_w} \PAR \ldots
% \EMPTYMVAR{c_a} \PAR \EMPTYMVAR{chk_a}
}
\\
 \multicolumn{4}{l}{ \PAR~ \bfdo \{\takeCh_x;
   \putS_z~a; \takeCh_z;\ldots \}} \\   %\ttakeMVar ~\istop
  \multicolumn{4}{l}{ \PAR~ \bfdo \{
   \putS_w~a; \takeCh_w; \ldots\}} \\   %\ttakeMVar ~\istop
   \multicolumn{4}{l}{\PAR~ \bfdo \{
   w\leftarrow\takeS_z; \ldots \}} \\   %\tputMVar ~(\tgetcheck~y); \treturn~()
 \end{array}$$
 Now the only  possible step is 
   $$\begin{array}{@{}l@{\,}l@{\,}l@{\,}l@{\,}}
\multicolumn{4}{@{}l}{\nu \ldots.x = \tChan~c_x~chk_x \PAR z = \tChan~c_z~chk_z\PAR w=\tChan~c_w~chk_w \PAR \ldots
}\\
\multicolumn{4}{l}{\PAR \MVAR{c_x}{w} \PAR \EMPTYMVAR{chk_x} \PAR \EMPTYMVAR{c_z} \PAR \EMPTYMVAR{chk_z} \PAR \MVAR{c_w}{a} \PAR  \EMPTYMVAR{chk_w} \PAR \ldots
% \EMPTYMVAR{c_a} \PAR \EMPTYMVAR{chk_a}
}
\\
 \multicolumn{4}{l}{ \PAR~ \bfdo \{\takeCh_x;
   \putS_z~a; \takeCh_z;\ldots \}} \\   %\ttakeMVar ~\istop
  \multicolumn{4}{l}{ \PAR~ \bfdo \{
    \takeCh_w; \ldots\}} \\   %\ttakeMVar ~\istop
   \multicolumn{4}{l}{\PAR~ \bfdo \{
   w\leftarrow\takeS_z; \ldots \}} \\   %\tputMVar ~(\tgetcheck~y); \treturn~()
 \end{array}$$
 and now the process is stuck.
%  $\begin{array}{@{\qquad}l@{\,}l@{\,}l@{\,}l@{\,}}
% %% \multicolumn{4}{l}{ \nu x,\tsendx,\tcheckx,z,\tsendz,\tcheckz,w,\tsendw,\tcheckw,a,\tsenda,\tchecka.}\\
%  (\MVAR{sendx}{w} &\PAR \EMPTYMVAR{checkx}  &\PAR \EMPTYMVAR{sendz} &\PAR \EMPTYMVAR{checkz} \\
%  ~ \EMPTYMVAR{sendw} &\PAR \EMPTYMVAR{checkw} &\PAR \EMPTYMVAR{senda} &\PAR \EMPTYMVAR{checka} \\
%  \multicolumn{4}{l}{ \PAR~ \bfdo \{\ttakeMVar~ \tcheckx;
%    \tputMVar ~\tsendz~ a; \ttakeMVar ~\tcheckz;\ldots \}} \\   %\ttakeMVar ~\istop
%   \multicolumn{4}{l}{ \PAR~ \bfdo \{ 
%    \tputMVar ~\tsendw~ a; \ttakeMVar ~\tcheckw; \ldots\}} \\   %\ttakeMVar ~\istop
%    \multicolumn{4}{l}{\PAR~ \bfdo \{ 
%    \ttakeMVar ~(\tgetsend~\tsendz); \ldots \}} \\   %\tputMVar ~(\tgetcheck~y); \treturn~()
%  \end{array}$\\
%  Now the process is stuck.  This can be prevented by our translation method to make the check-MVar private.
\end{remark}
}

\begin{proposition}\label{prop:fourtrans-executable}
  The translations $\mathfrak{T}_1,\mathfrak{T}_2,\mathfrak{T}_3$, and $\mathfrak{T}_4$ are executable, communicating, and overlap-free, whereas the translations $\mathfrak{T}_5$ and $\mathfrak{T}_6$  are executable, communicating, 
  but overlapping.
\end{proposition}
\begin{proof}
% Executability and overlap-freeness can 
% easily be checked by all executions. 
We only consider overlaps. For $\mathfrak{T}_1$ - $\mathfrak{T}_4$, only if all 8 actions are finished,
the next send or receive can start.
For $\mathfrak{T}_5,\mathfrak{T}_6$,
after executing $\putCh^1$, $\takeCh^1$,  we can again execute $\putCh^1$.
\end{proof}

% Let $\mathfrak{T}_1= (T_{\mathit{send}},T_{\mathit{receive}})=([ \putS,\putCh^1,\takeCh^2,\putCh^3 ],[
% \takeCh^1,\putCh^2,\takeCh^3,\takeS ])$
% the first gstb-translation from \cref{table-3mvar}.
In \cite{schmidt-schauss-sabel:frank-60:19} we argue
that the induced translation $\tauInduced[\mathfrak{T}_1]$
leaves may- and should-convergence invariant.
The main help in reasoning is that there is no unintended interleaving of send and receive sequences according to \cref{prop:fourtrans-executable}.
Application of \cref{prop:adequateInduced} then shows:
% the following theorem.
\begin{theorem}\label{thm:tauinduced1-correct}
Translation $\tauInduced[\mathfrak{T}_1]$ is adequate.
%%  and on closed processes it is fully-abstract.
\end{theorem}

For 4 MVars, our tool refutes 9266  and there remain 334 candidates for correct translations.

\subsection{Dropping the Interprocess Check Restriction}
We now consider gstb-translations without the interprocess check restriction, i.e.~$\putCh^i$ 
 and $\takeCh^i$ both may occur in the sender-program (or the receiver program, resp.).
 If we allow one check-MVar without reuse,
 then there are 20 candidates for translations.
 All are refuted by our simulation.  %   since counter-examples are found. 
%  A new counter-example which was not used in the cases before is the 
%  process $\outputxy{x}{y} \PAR \inputxy{x}{y}.\tStop$
%   which is may- and should-convergent and refutes the correctness of
%  the translations     $([\putS],[\takeCh^1,\takeS,\putCh^1])$
%  and   $([\putS],[\takeS,\takeCh^1,\putCh^1])$, since after the translation the receiver 
%  deadlocks at the $\takeCh^1$-operation and thus the programs is neither may- nor should-convergent.
% 
 Allowing reuse of the single check-MVar seems not to help to construct a correct translation:
 We simulated this for up to 6 uses, leading to 420420 candidates for a correct translation 
 -- our simulation refutes all of them.
%   1 20          
%   2 210     
%   3 1680 
%   4 11550
%   5 72072
%   6 420420
%  
 \begin{conjecture}\label{conjecture}
We conjecture that there is no correct translation for one check-MVar where 
 re-uses are permitted and the interprocess check restriction is dropped,
 i.e., $T_{send}$ is a word over $\{\putS,\putCh,$ $\takeCh\}$ and $T_{receive}$ a word over $\{\takeS,\putCh,$ $\takeCh\}$, where $\putS,\takeS$
 occur exactly once.
 \footnote{We already have a proof in the meantime,  not yet published.}
 \end{conjecture}
For two MVars, one use and without the interprocess check restriction there are 420 translations. Our tool refutes all except for two: $\mathfrak{T}_7=([\putCh^1,\putS,\takeCh^2,\takeCh^1],[\takeS,\putCh^2])$ and 
$\mathfrak{T}_8=([\takeCh^1,\putS],
[\putCh^2,\putCh^1,\takeS,\takeCh^2])$.
% are not refuted by our simulation.
In $\mathfrak{T}_7$ the second check-MVar is used as a mutex for the senders, ensuring that only one sender can operate at a time.
$\mathfrak{T}_8$ does the same on the receiver side.

\begin{proposition}\label{prop:2MVars-non-interprocess-check-restriction}
  The translations $\mathfrak{T}_7,\mathfrak{T}_8$
   are executable, communicating, overlap-free.
\end{proposition}
\begin{proof}
%\davidins{}
% It is easy to verify that 
The translations are executable and communicating. 
For $\mathfrak{T}_7$, $\putCh^1,\putS$ and $\takeS$ are
performed in this order.
An additional sender
cannot execute its first command before the original sender performs $\takeCh^1$ and this again is only possible after the receiver program is finished. An additional receiver can only be executed after a $\putS$ is performed, which cannot be done by the current sender and receivers.
For $\mathfrak{T}_8$,  $\putCh^2,\putCh^1$ and $\takeCh^1$ are performed in this order. An additional receiver can only start after $\takeCh^2$ was executed by the original receiver, which can only occur after the original sender and receiver program are fully evaluated. An additional sender  can only start after $\putCh^1$ has been executed again, but the current sender and receiver cannot execute this command.
\end{proof}
The induced translation $\tauInduced[\mathfrak{T}_7]$
% from $\mathfrak{T}_7$
% the first gstb-translation, \ie
% %\davidins{..}
% the translation 
% 
% with 
% $$\mathfrak{T}_2 = ([\putCh^1,\putS,\takeCh^2,\takeCh^1],[\takeS,\putCh^2])$$
is (closed) 
convergence-equivalent  \cite{schmidt-schauss-sabel:frank-60:19}.
With \cref{prop:adequateInduced} this shows:
\begin{theorem}\label{thm:tauinduced2-correct}
Translation $\tauInduced[\mathfrak{T}_7]$ is adequate.
%%  and on closed processes it is fully-abstract.
\end{theorem}

We are convinced that the same holds for $\mathfrak{T}_8$.
We conclude the statistics of our search for translations without the interprocess restriction:
%  \begin{itemize}
% \item 
For 3 MVars, there are 10080 translations and 9992 are refuted, \ie~98 are potentially correct. One is
% Among them, there is the quite intuitive translation
$([\putCh^1,\putS,\takeCh^2,\takeCh^1],[\putCh^3,\takeS,\putCh^2,\takeCh^3])$
which is quite intuitive: 
check-MVar $1$ is used as a mutex for all senders on the same channel, check-MVar $3$ is used as a mutex for all receivers, and check-MVar $2$ is used to send an acknowledgement.
% that the message was received.
% \item 
% 
% 
For 4 MVars, there are  277200  translations and 273210 are refuted, \ie~3990  are potentially correct.
\section{Discussion and Conclusion}\label{sec:concl}
We investigated translating 
the $\pi$-calculus into $\CH$ and showed correctness and adequacy 
of a translation $\tau_0$ with private MVars for every 
translated communication. 
% The translation $\tau_0$ is adequate \wrt~the respective contextual equivalences. 
% This supports our claim that the contextual semantics using may- and should-convergence is a useful criterion for concurrent programs. 
% 
For translations with global names, 
we started an investigation on exhibiting (potentially) correct translations. 
% In the design space of all (global) translations 
We identified several minimal potentially correct translations
and characterized all incorrect ``small'' translations. 
%   where we left the issue open whether there is a correct translation with only one check-MVar.
%   with unrestricted usage.
% 
For two particular global translations, we have shown  that they are convergence-equivalent
%for 
% preserve and reflect the 
% for may- and should-convergence on closed processes.
and we proved their adequacy on open processes.  %%\manfredins  %% (on open processes) and full-abstraction for closed processes.
The exact form of the translations were found
by our tool to search for translations
and to refute their correctness. The tool showed that there is no correct gstb-translation with the interprocess check restriction for less than 3 check-MVars.
%  
% As further work we are working on proving Conjecture~\ref{conjecture}\footnote{We already have a proof in the meantime, but it is not published at the moment.}
% \manfred{footnote dass wir es bewiesen haben ??}
We also may consider 
extended variants of the $\pi$-calculus. We are convinced that adding recursion and sums can easily be built into 
the translation, while it might be challenging to encode  mixed sums or (name) matching operators. For name matching operators, our current translation
would require to test usual bindings in $\mathit{CH}$ for equality which is not available in core-Haskell. 
Solutions may either use an adapted translation or a target language that
supports observable sharing \cite{Claessen-Sands:99,Gill:09}.  The translation of mixed-sums into $\mathit{CH}$ appears to require  more complex translations,
where the send- and receive-parts are not linear lists of actions.
% Another research question is whether the $\pi$-calculus can be embedded in call-by-need  lambda calculi with
%  {\tt amb}, where
% in \cite{carayol-hirschkoff-sangiorgi:05} it is shown that {\tt amb} is encodable in the $\pi$-calculus. 
% 
% 
%%% TEXEXPAND: END FILE ./conclusion.tex
%%% TEXEXPAND: END FILE ./pitrans-text.tex
%%%%%%%%%%%%%%%%%%%% 
\paragraph*{Acknowledgments}
% \acknowledgments
We thank the anonymous reviewers for their valuable comments. In particular, we thank an anonymous reviewer for advises to improve  the construction of translations, and  for providing the counter-example
in the last row of \cref{table-single-check}.
\bibliographystyle{eptcs}
\bibliography{pichfbib}
\clearpage
% \appendix 
% \input{pitrans-appendix.tex}
\end{document}